\documentclass{article}
\usepackage{fullpage}

\usepackage[utf8]{inputenc}
\usepackage{amsthm,amssymb}
\usepackage{amsmath}
\usepackage{thmtools,mathtools}
\usepackage{caption,subcaption}
\usepackage{epsfig,graphicx,graphics,color,tikz,tikz-cd}
\usepackage [linesnumbered, ruled,vlined]{algorithm2e}
\usetikzlibrary{calc,decorations.pathmorphing,shapes}
\usepackage[most]{tcolorbox}
\usepackage{array,makecell,multirow}
\setcellgapes{4pt}%parameter for the spacing
\usepackage{bbm}
\usepackage{varwidth}
\usepackage{rotating}
\usepackage{booktabs}
\usepackage[mathcal,mathscr]{eucal}
\usepackage{varwidth}
\usepackage{enumerate,enumitem}
\usepackage[sort]{cite}
\usepackage{xspace}
\usepackage{bm}
\usepackage{csquotes}
\usepackage{hyperref}
\hypersetup{
colorlinks=true,       % false: boxed links; true: colored links
linkcolor=blue,        % color of internal links (change box color with linkbordercolor)
citecolor=magenta,         % color of links to bibliography
filecolor=magenta,     % color of file links
urlcolor=cyan,         % color of external links
linktoc=all
}

\theoremstyle{plain}
\newtheorem{thm}{Theorem}[section]
\newtheorem{lem}[thm]{Lemma}
\newtheorem{cor}[thm]{Corollary}
\newtheorem{cl}[thm]{Claim}

\theoremstyle{definition}

\newtheorem{rem}[thm]{Remark}

\def\final{0}  % set this to 1 to get a comment-free version
\def\iflong{\iffalse}
\ifnum\final=0  %namely if we allow comments in the output
\newcommand{\knote}[1]{{\color{red}[{\tiny \textbf{Kristóf:} \bf #1}]\marginpar{\color{red}*}}}
\newcommand{\tnote}[1]{{\color{blue}[{\tiny \textbf{Tamás:} \bf #1}]\marginpar{\color{blue}*}}}
\newcommand{\gnote}[1]{{\color{orange}[{\tiny \textbf{Gergő:} \bf #1}]\marginpar{\color{orange}*}}}
\newcommand{\ynote}[1]{{\color{purple}[{\tiny \textbf{Yuhang:} \bf #1}]\marginpar{\color{purple}*}}}
\else % in this case [final=1] we don't want any comments to show
\newcommand{\knote}[1]{}
\newcommand{\tnote}[1]{}
\newcommand{\gnote}[1]{}
\newcommand{\ynote}[1]{}
\fi

\makeatletter
\renewcommand{\paragraph}{%
  \@startsection{paragraph}{4}%
  {\z@}{1.5ex \@plus 1ex \@minus .2ex}{-1em}%
  {\normalfont\normalsize\bfseries}%
}
\makeatother

\DeclareMathOperator\comp{comp}

\newcommand{\bR}{\mathbb{R}}

\newcommand{\bZ}{\mathbb{Z}}

\newcommand{\cI}{\mathcal{I}}

\newcommand{\cM}{{\bm{M}}}

\newcommand{\maxpf}{\textsc{Max-PF}\xspace}
\newcommand{\maxpt}{\textsc{Max-PT}\xspace}
\newcommand{\pst}{\textsc{PST}\xspace}
\newcommand{\minbdst}{\textsc{Min-BDST}\xspace}
\newcommand{\minubdst}{\textsc{Min-UBDST}\xspace}
\newcommand{\bdst}{\textsc{BDST}\xspace}
\newcommand{\ubdst}{\textsc{UBDST}\xspace}
\newcommand{\opt}{\textsc{Opt}\xspace}
\newcommand{\maxlf}{\textsc{Max-LF}\xspace}
\newcommand{\tsp}{\textsc{$(1,2)$-TSP}\xspace}
\newcommand{\longP}{\textsc{Longest-Path}\xspace}
\newcommand{\longUP}{\textsc{Undirected-Longest-Path}\xspace}

\definecolor[named]{Blue}{cmyk}{1,0.1,0,0.1}
\definecolor[named]{Yellow}{cmyk}{0,0.16,1,0}
\definecolor[named]{Orange}{cmyk}{0,0.42,1,0.01}
\definecolor[named]{Red}{cmyk}{0,0.90,0.86,0}
\definecolor[named]{LightBlue}{cmyk}{0.49,0.01,0,0}
\definecolor[named]{Green}{cmyk}{0.20,0,1,0.19}

\SetKwFor{Whileblock}{while}{do}{}
\SetKwBlock{Begin}{}{}
\SetVlineSkip{1.5pt}
\makeatletter
\newcommand{\linkdest}[1]{\Hy@raisedlink{\hypertarget{#1}{}}}
\makeatother

\newlength{\bibitemsep}
\newlength{\bibparskip}
\let\oldthebibliography\thebibliography
\renewcommand\thebibliography[1]{%
\oldthebibliography{#1}%
\setlength{\parskip}{\bibitemsep}%
\setlength{\itemsep}{\bibparskip}%
}

\title{Approximating maximum-size properly colored forests}

\author{
Yuhang Bai\thanks{School of Mathematics and Statistics, Northwestern Polytechnical University and Xi'an-Budapest Joint Research Center for Combinatorics, Xi'an 710129,
Shaanxi, People's Republic of China. Email: \texttt{yhbai@mail.nwpu.edu.cn}.}
\and
Kristóf Bérczi\thanks{MTA-ELTE Matroid Optimization Research Group and HUN-REN–ELTE Egerváry Research Group, Department of Operations Research, ELTE Eötvös Loránd University, Budapest, Hungary. Email: \texttt{kristof.berczi@ttk.elte.hu,tamas.schwarcz@ttk.elte.hu}.}
\and
Gergely Csáji\thanks{Department of Operations Research, ELTE Eötvös Loránd University, Budapest, Hungary and HUN-REN Centre for Economic and Regional Studies Hungary. Email: \texttt{csaji.gergely@krtk.hun-ren.hu}.}
\and
Tamás Schwarcz\footnotemark[2]
}

\date{}

\begin{document}
\maketitle
\thispagestyle{empty}
%%%%%%%%%%%%%%%%%%%%%%%%%%%%%%%%

%%%%%%%%%%%%%%%%%%%%%%%%%%%%%%%%
\begin{abstract}

In the \emph{Properly Colored Spanning Tree} problem, we are given an edge-colored undirected graph and the goal is to find a properly colored spanning tree, i.e., a spanning tree in which any two adjacent edges have distinct colors. The problem is interesting not only from a graph coloring point of view, but is also closely related to the \emph{Degree Bounded Spanning Tree} and \emph{$(1,2)$-Traveling Salesman} problems, two classical questions that have attracted considerable interest in combinatorial optimization and approximation theory. Previous work on properly colored spanning trees has mainly focused on determining the existence of such a tree and hence has not considered the question from an algorithmic perspective. We propose an optimization version called \emph{Maximum-size Properly Colored Forest} problem, which aims to find a properly colored forest with as many edges as possible. We consider the problem in different graph classes and for different numbers of colors, and present polynomial-time approximation algorithms as well as inapproximability results for these settings. Our proof technique relies on the sum of matching matroids defined by the color classes, a connection that might be of independent combinatorial interest. 

We also consider the \emph{Maximum-size Properly Colored Tree} problem, which asks for the maximum size of a properly colored tree not necessarily spanning all the vertices. We show that the optimum is significantly more difficult to approximate than in the forest case, and provide an approximation algorithm for complete multigraphs.
\medskip

\end{abstract}
%%%%%%%%%%%%%%%%%%%%%%%%%%%%%%%%
%%%%%%%%%%%%%%%%%%%%%%%%%%%%%%%%

\newpage
\pagenumbering{roman}
\tableofcontents
\newpage
\pagenumbering{arabic}
\setcounter{page}{1}
%%%%%%%%%%%%%%%%%%%%%%%%%%%%%%%%

%%%%%%%%%%%%%%%%%%%%%%%%%%%%%%%%
\section{Introduction}
\label{sec:intro}
%%%%%%%%%%%%%%%%%%%%%%%%%%%%%%%%

Throughout the paper, we consider loopless graphs that might contain parallel edges. A \emph{$k$-edge-colored graph} is a graph $G=(V,E)$ with a coloring $c\colon E\to [k]$ of its edges by $k$ colors. We refer to a graph that is $k$-edge-colored for some $k\in\bZ_+$ as \emph{edge-colored}.  A subgraph $H$ of $G$ is called \emph{rainbow colored} if no two edges of $H$ have the same color, and \emph{properly colored} if any two adjacent edges of $H$ have distinct colors. Since rainbow colored forests form the common independent sets of two matroids, i.e., the partition matroid defined by the color classes and the graphic matroid of the graph, a rainbow colored forest of maximum size can be found in polynomial time using Edmonds' celebrated matroid intersection algorithm~\cite{edmonds1970submodular}. However, much less is known about the properly colored case. In~\cite{borozan2019maximum}, Borozan, de La Vega, Manoussakis, Martinhon, Muthu, Pham, and Saad initiated the study of properly edge-colored spanning trees of edge-colored graphs and investigated the existence of such a spanning tree, called the \emph{Properly Colored Spanning Tree} problem (\pst). This problem generalizes the well-known bounded degree spanning tree problem for uncolored graphs as the number of colors bounds the degree of each vertex, as well as the properly colored Hamiltonian path problem when the number of colors is restricted to two. Since both of these problems are NP-complete, finding a properly colored spanning tree is hard in general. 

The aim of this paper is to study the problem from an approximation point of view. Accordingly, we define the \emph{Maximum-size Properly Colored Forest} problem (\maxpf) in which the goal is to find a properly colored forest of maximum size in an edge-colored graph, and discuss the approximability of the problem in various settings. Throughout the paper, by the \emph{size} of a tree or a forest we mean the number of its edges.

%%%%%%%%%%%%%%%%%%%%%%%%%%%%%%%%
\subsection{Related work and connections}
\label{sec:related}
%%%%%%%%%%%%%%%%%%%%%%%%%%%%%%%%

Finding properly colored spanning trees in graphs is closely related to constrained spanning tree problems, or in a more general context, to the problem of finding a basis of a matroid subject to further matroid constraints. In what follows, we give an overview of questions that motivated our investigations. 

%%%%%%%%%%%%%%%%%%%%%%%%%%%%%%%%
\paragraph{Properly colored trees.}
%%%%%%%%%%%%%%%%%%%%%%%%%%%%%%%%

Properly colored spanning trees were first considered in Borozan et al.~\cite{borozan2019maximum} where their existence was studied from both a graph-theoretic and an algorithmic perspective. They showed that finding a properly colored spanning tree remains NP-complete when restricted to complete graphs. They also proved that \maxpt is hard to approximate within a factor of $55/56+\varepsilon$ for any $\varepsilon>0$, and conjectured the existence of an approximation algorithm with logarithmic approximation guarantee. On the positive side, they provided polynomial algorithms for graphs not containing properly edge-colored cycles. 

Since deciding the existence of a properly colored spanning tree is hard in general, a considerable amount of work has focused on finding sufficient conditions. For an edge-colored graph $G=(V,E)$, the \emph{color degree} of a vertex is the number of distinct colors appearing on the edges incident with it. Let $\delta^c(G)$ denote the minimum value among the color degrees of all the vertices of $G$, called the \emph{minimum color degree} of $G$. As a counterpart of Dirac's theorem~\cite{dirac1952some} on Hamiltonian cycles, Cheng, Kano and Wang~\cite{cheng2020properly} verified that the condition $\delta^c(G) \geq |V|/2$ implies the existence of a properly colored spanning tree. This result was further strengthened by Kano, Maezawa, Ota, Tsugaki, and Yashima~\cite{kano2020color} who proved an analogous result under the assumption that the color degree sum of any two adjacent vertices of $G$ is at least $|V|$. For bipartite graphs, Kano and Tsugaki~\cite{kano2021rainbow} showed that the lower bound on the minimum color degree can be improved to $\delta^c(G) \geq |V|/3+1$ while still implying the existence of a properly colored spanning tree.

Since a properly colored spanning tree may not exist, it is natural to ask for the maximum size of a properly colored tree not necessarily spanning all the vertices, called the \emph{Maximum-size Properly Colored Tree} problem (\maxpt). Hu, Liu and Maezawa~\cite{hu2022maximum} proved that the maximum size of a properly colored tree in an edge-colored connected graph is at least $\min \left\{|V|-1, 2 \delta^c(G)-1\right\}$. Note that this implies the existence of a properly colored spanning tree whenever $\delta^c(G)\geq |V|/2$, and so it generalizes the result of~\cite{cheng2020properly}.

%%%%%%%%%%%%%%%%%%%%%%%%%%%%%%%%
\paragraph{Degree bounded spanning trees.}
%%%%%%%%%%%%%%%%%%%%%%%%%%%%%%%%

In the \emph{Minimum Bounded Degree Spanning Tree} problem (\minbdst), we are given an undirected graph $G=(V,E)$ with $|V|=n$, a cost function $c\colon E\to\bR$ on the edges, and degree upper bounds $g\colon V\to\bZ_+$ on the vertices, and the task is to find a spanning tree of minimum cost that satisfies all the degree bounds. When the degree bounds are the same for every vertex, we get the \emph{Minimum Uniformly Bounded Degree Spanning Tree} problem (\minubdst). Furthermore, we drop ``\textsc{Min}'' from the notation when the edge-costs are identically one, thus leading to problems \bdst and \ubdst. For ease of discussion, below it is always assumed that the problem in question has a feasible solution, and for optimization problems the optimum value is denoted by \opt.

For the \ubdst problem with the upper bounds being identically $k\in\bZ_+$, F\"urer and Raghavachari~\cite{furer1992approximating} gave an iterative polynomial time algorithm that computes a spanning tree of maximum degree at most $O(k+\log n)$. Their algorithm extends to the Steiner case as well, that is, when only some subset of vertices need to be connected. They also described a refined version of their algorithm that produces a spanning tree of maximum degree $k+1$ and observed that, unless $P=NP$, this is the best bound achievable in polynomial time. Czumaj and Strothmann~\cite{czumaj1997bounded} studied the problem under additional connectivity and maximum degree assumptions. They presented algorithms that find a spanning tree of maximum degree at most $2k-2$ in biconnected graphs, and a spanning tree of maximum degree $k$ in $\lambda$-connected graphs of maximum degree $\lambda(k-2)+2$. On the hardness side, they proved that it is NP-complete to decide whether a $\lambda$-connected graph of maximum degree $\lambda(k-2)+3$ has a spanning tree of maximum degree $k$, provided $k\leq 3$, and that the same result holds for $\lambda$-connected graphs of maximum degree $\lambda(k-1)$ if $k\geq 3$.

Fischer~\cite{fischer1993optimizing} observed that the algorithms of~\cite{furer1992approximating} can be adapted to find a minimum cost spanning tree of approximate maximum degree. Let us emphasize that in this problem, the roles of the edge-costs and the degree bounds are reversed compared to the \ubdst problem: here the objective is to minimize the maximum degree of the solution over minimum cost spanning trees, while in the \ubdst problem the objective is to minimize the cost of the solution over spanning trees of maximum degree at most $k$. Fischer showed that if there exists a minimum cost spanning tree of maximum degree $k$, then one can find a minimum cost spanning tree of maximum degree at most $O(k+\log n)$ in polynomial time. By relying on the refined version of the algorithm of~\cite{furer1992approximating}, he also presented an algorithm to construct a minimum cost spanning tree of degree at most $\ell(k+1)$, where $\ell$ is the number of different cost values on the edges in the graph. 

The \minubdst problem was first proposed by Ravi, Marathe, Ravi, Rosenkrantz, and Hunt III~\cite{ravi1993many} in the context of network design problems with multiple design objectives. Using the local search technique of Fischer~\cite{fischer1993optimizing}, they  gave a polynomial algorithm that finds a spanning tree of maximum degree $O(k\log (n/k))$ and of cost at most $O(\log (n/k)\opt)$. It is worth mentioning that their techniques generalize to the case of constructing Steiner trees or generalized Steiner forests as well. Using ideas from Lagrangean duality, K\"onemann and Ravi~\cite{konemann2002matter} improved this by describing an algorithm that finds a spanning tree of maximum degree $O(k+\log n)$ and cost $O(\opt)$. The first approximation algorithms that find trees of optimal cost is due to Chaudhuri, Rao, Riesenfeld and Talwar~\cite{chaudhuri2005would,chaudhuri2009would}, who presented two algorithms. The first one runs in polynomial time and finds a tree of cost at most $\opt$ with maximum degree at most $kb/(2-b)+O(\log_b(n))$ for any $b\in(1,2)$. The second algorithm has quasi-polynomial running time and produces a tree of cost at most $\opt$ and maximum degree at most $k+O(\log n/\log\log n)$. A peculiar feature of the algorithms is that they involve lower bound requirements on the degrees, therefore the results extend to a generalized version of the problem in which both upper and lower degree bounds are given. The breakthrough result of Goemans~\cite{goemans2006minimum} showed that a spanning tree of maximum degree at most $k+2$ and of cost at most $\opt$ can be found in polynomial time, and formulated an analogous statement with $k+1$ instead of $k+2$ as a conjecture. 

Using the iterative rounding method introduced by Jain~\cite{jain2001factor}, Singh and Lau~\cite{singh2015approximating} verified Goemans' conjecture for the more general \minbdst problem. Their algorithm extends to the setting when both lower and upper bounds are given on the degree of each vertex.

%%%%%%%%%%%%%%%%%%%%%%%%%%%%%%%%
\paragraph{Degree bounded matroids and multi-matroid intersection.}
%%%%%%%%%%%%%%%%%%%%%%%%%%%%%%%%

Kir\'aly, Lau and Singh~\cite{kiraly2012degree} studied a matroidal extension of the \minbdst problem. In their setting, a matroid with a cost function on its elements, and a hypergraph on the same ground set with lower and upper bounds $f(e)\leq g(e)$ for each hyperedge $e$. The task is to find a minimum cost basis of the matroid which contains at least $f(e)$ and at most $g(e)$ elements from each hyperedge $e$. The algorithm presented in~\cite{kiraly2012degree} is also based on the iterative rounding technique of~\cite{jain2001factor}, and determines a basis of cost at most \opt and violating the degree bounds by at most $2\Delta-1$ if both lower and upper bounds are present, and by at most $\Delta-1$ if only lower or only upper bounds are given. If we choose the matroid to be the graphic matroid of a graph $G=(V,E)$ and the hyperedges to be the sets $\delta(v)$ for $v\in V$, we get back the \minbdst problem with the value of $\Delta$ being $2$.

In~\cite{zenklusen2012matroidal}, Zenklusen considered a different generalization of the \minbdst problem where for every vertex $v$, the edges adjacent to $v$ have to be independent in a matroid $\cM_v$. He presented an algorithm that returns a spanning tree of cost at most \opt, such that for every vertex $v$, it suffices to remove at most $8$ edges from the spanning tree to satisfy the matroidal degree constraint at $v$. This model was further extended by Linhares, Olver, Swamy and Zenklusen~\cite{linhares2020approximate} who studied the problem of finding a minimum cost basis of a matroid $\cM_0$ that is independent in other matroids $\cM_1,\dots,\cM_q$. They derived an algorithm that is based on an iterative refinement technique and returns a solution that violates the rank constraints by a multiplicative factor which depends on how strongly the ground sets of the matroids $\cM_i$ overlap.

%%%%%%%%%%%%%%%%%%%%%%%%%%%%%%%%
\paragraph{$(1,2)$-traveling salesman problem.}
%%%%%%%%%%%%%%%%%%%%%%%%%%%%%%%%

The metric Traveling Salesman Problem is one of the most fundamental combinatorial optimization problems. Karp~\cite{karp1972reducibility} showed that the problem is NP-hard even in the special case when all distances between cities are either $1$ or $2$, called the \emph{Traveling Salesman Problem with Distances 1 and 2} (\tsp). This result was further strengthened by Papadimitriou and Yannakakis~\cite{papadimitriou1993traveling} who showed that \tsp is in fact hard to approximate and MAX-SNP-hard. The currently best known inapproximability bound of $535/534$ is due to Karpinski and Schmied~\cite{karpinski2012approximation}. From an approximation point of view, the most natural local search approximation algorithm starts with an arbitrary tour and replaces at most $k$ edges in every iteration to get a shorter tour, or outputs the tour if no such improvement can be found. Khanna, Motwani, Sudan and Vazirani~\cite{khanna1998syntactic} gave an upper bound of $3/2$ on the approximation guarantee of this method for $k=2$, while Zhong~\cite{zhong2020approximation} showed a lower bound of $11/10$ for any fixed $k$. A local search-based $8/7$-approximation algorithm with running time $O(n^9)$ was given by Berman and Karpinski~\cite{berman20068}, while Adamaszek, Mnich and Paluch~\cite{adamaszek2018new} presented a faster algorithm with running time $O(n^3)$ achieving the same approximation factor. 

\bigskip

The problem \maxpf is closely related to the problems listed above. 
\begin{itemize}
\item \maxpf provides a relaxation of both the \pst and \maxpt problems. 
\item For an arbitrary graph $G$, let $G'$ be the $k$-edge-colored multigraph obtained by taking $k$ copies of each edge of $G$ colored by different colors. Then, $G$ has a uniformly bounded degree spanning tree with upper bound $k$ if and only if $G'$ admits a properly colored spanning tree.
\item For a $k$-edge-colored graph $G=(V,E)$, let $M$ be the graphic matroid of $G$. Furthermore, define a hypergraph on $E$ as follows: for each vertex $v\in V$ and color $i\in[k]$, let $e_{v,i}\coloneqq\{e\in E\mid c(e)=i,e\ \text{is incident to $v$}\}$ be a hyperedge with upper bound $1$. Then, $G$ has a properly colored spanning tree if and only if $M$ admits a degree bounded basis.
\item For a $k$-edge-colored graph $G=(V,E)$, let $M_0$ be the graphic matroid of $G$. Furthermore, for each vertex $v\in V$ and color $i\in[k]$, let $M_{v,i}$ be a rank-$1$ partition matroid whose ground set is the set of edges incident to $v$ having color $i$. Then, $G$ has a properly colored spanning tree if and only if the multi-matroid intersection problem $M_0,\{M_{v,i}\}_{v\in V,i\in[k]}$ admits a solution of size $|V|-1$.
\item  Consider an instance of \tsp on $n$ vertices and let $G$ denote the subgraph of edges of length $1$. Since any linear forest of $G$ of size $x$ can be extended to a Hamiltonian cycle of length $2n-x$, one can reformulate \tsp as the problem of finding a maximum linear forest in $G$. This problem reduces to \maxpf in $2$-edge-colored graphs, see Section~\ref{sec:hardnesspf} for further details. 
\end{itemize}

Consider now an instance of \maxpf, that is, an edge-colored graph $G$ and let \opt denote the maximum size of a properly colored forest in $G$. One can obtain a forest $F$ of $G$ of size at least \opt in which every color appears at most twice at every vertex, either by the approximation algorithm of~\cite{kiraly2012degree} for the bounded degree matroid problem, or by the approximation algorithm of~\cite{linhares2020approximate} for the multi-matroid intersection problem. Deleting conflicting edges from $F$ greedily results in a properly colored forest of size at least $|F|/2\geq\opt/2$, thus leading to a $1/2$-approximation for \maxpf. Our main motivation was to improve the approximation factor and to understand the inapproximability of the problem.

%%%%%%%%%%%%%%%%%%%%%%%%%%%%%%%%
\subsection{Our results}
\label{sec:results}
%%%%%%%%%%%%%%%%%%%%%%%%%%%%%%%%

We use the convention that, by an \emph{$\alpha$-approximation algorithm}, for minimization problems we mean an algorithm that provides a solution with objective value at most $\alpha$ times the optimum for some $\alpha\geq 1$, while for maximization problems we mean an algorithm that provides a solution with objective value at least $\alpha$ times the optimum for some $\alpha\leq 1$.

We initiate the study of properly colored spanning trees from an optimization point of view and focus on the problem of finding a properly colored forest of \emph{maximum size}, i.e., containing a maximum number of edges. We discuss the problem for several graph classes and numbers of colors, and provide approximation algorithms as well as inapproximability bounds for these problems. The results are summarized in Table~\ref{table:results}.

\begin{table}[h!]
\centering
\renewcommand*{\arraystretch}{1.2}
\setlength{\tabcolsep}{2pt}
\newcommand{\tableentrybas}[2]{\colorbox{#1}{\parbox[c][0.9em][c]{122mm}{\centering\footnotesize #2}}}
\newcommand{\tableentrybass}[2]{\colorbox{#1}{\parbox[c][0.9em][c]{80mm}{\centering\footnotesize #2}}}
\newcommand{\tableentrybase}[2]{\colorbox{#1}{\parbox[c][2.8em][c]{18mm}{\centering\footnotesize #2}}}
\newcommand{\tableentrybasee}[2]{\colorbox{#1}{\parbox[c][0.9em][c]{38mm}{\centering\footnotesize #2}}}
\newcommand{\tableentrybaseeb}[2]{\colorbox{#1}{\parbox[c][2.9em][c]{38mm}{\centering\footnotesize #2}}}
\newcommand{\tableentrybaseee}[2]{\colorbox{#1}{\parbox[c][2.9em][c]{40mm}{\centering\footnotesize #2}}}
\newcommand{\tableentrybaseeee}[2]{\colorbox{#1}{\parbox[c][6em][c]{38mm}{\centering\footnotesize #2}}}
\footnotesize
\begin{tabular}{c|c|c|c}
& \multicolumn{3}{c}{Number of colors}\\ 
Graphs & $k=2$  & $k=3$ & $k\geq 4$ \\
\cmidrule(lr){1-1}\cmidrule(lr){2-2}\cmidrule(lr){3-3}\cmidrule(lr){4-4}\\[-1.6em]

\multirow[c]{2}{*}{\tableentrybase{Red!0}{Simple \linebreak graphs}} & 
\multicolumn{3}{c}{\tableentrybas{Red!20}{MAX-SNP-hard (Thm.~\ref{thm:2snp})}}\\             
& \tableentrybasee{Green!30}{$3/4$-approx.\ (Thm.~\ref{thm:34})} & 
\tableentrybasee{Green!30}{$5/8$-approx.\ (Thm.~\ref{thm:58})} & 
\tableentrybasee{Green!30}{$4/7$-approx.\ (Thm.~\ref{thm:simple})}\\ 
\cmidrule(lr){1-1}\cmidrule(lr){2-2}\cmidrule(lr){3-3}\cmidrule(lr){4-4}\\[-1.6em]

\multirow[c]{2}{*}{\tableentrybase{Red!0}{Multigraphs}} & 
\multicolumn{3}{c}{\tableentrybas{Red!20}{MAX-SNP-hard (Thm.~\ref{thm:2snp})}}\\

& \tableentrybasee{Green!30}{$3/5$-approx.\ (Thm.~\ref{thm:multi2})} & 
\tableentrybasee{Green!30}{4/7-approx.\ (Thm.~\ref{thm:simple})} & 
\tableentrybasee{Green!30}{$5/9$-approx.\ (Thm.~\ref{thm:main})}\\ 
\cmidrule(lr){1-1}\cmidrule(lr){2-2}\cmidrule(lr){3-3}\cmidrule(lr){4-4}\\[-1.6em]

% \multirow[c]{2}{*}{\tableentrybase{Red!0}{Complete \\multigraphs}} & 
% \multirow[c]{2}{*}[0.2em]{\tableentrybaseeb{Green!30}{P (Thm.~\ref{thm:ck2})}} & 
% \multicolumn{2}{c}{\tableentrybass{Red!20}{MAX-SNP-hard (Thm.~\ref{thm:3snp})}}\\
% & &\tableentrybasee{Green!30}{0.529-approx.\ (Thm.~\ref{thm:3cm})} & 
% \tableentrybasee{Green!30}{$6/11$-approx.\ (Thm.~\ref{thm:main})}

\multirow[c]{2}{*}{\tableentrybase{Red!0}{Complete \\graphs}} & 
\multirow[c]{4}{*}[0.1em]{\tableentrybaseeee{Green!30}{P (Thm.~\ref{thm:ck2})}} & 
\multicolumn{2}{c}{\tableentrybass{Red!20}{MAX-SNP-hard (Thm.~\ref{thm:3snp})}}\\
& & \tableentrybasee{Green!30}{$5/8$-approx.\ (Thm.~\ref{thm:58})} & 
\tableentrybasee{Green!30}{$4/7$-approx.\ (Thm.~\ref{thm:simple})}\\
\cmidrule(lr){1-1}\cmidrule(lr){3-3}\cmidrule(lr){4-4}\\[-1.6em]

\multirow[c]{2}{*}{\tableentrybase{Red!0}{Complete \\multigraphs}} & 
& 
\multicolumn{2}{c}{\tableentrybass{Red!20}{MAX-SNP-hard (Thm.~\ref{thm:3snp})}}\\
& & \tableentrybasee{Green!30}{4/7-approx.\ (Thm.~\ref{thm:simple})} & 
\tableentrybasee{Green!30}{$5/9$-approx.\ (Thm.~\ref{thm:main})}
\end{tabular}
\caption{Complexity landscape of \maxpf. \label{table:results}}
\end{table}

We also consider \maxpt, that is, when a properly colored tree (not necessarily spanning) of maximum size is to be found. We give a strong inapproximability result in general, together with an approximation algorithm for complete multigraphs. The results are summarized in Table~\ref{table:MaxPTresults}.

\begin{table}[h!]
\centering
\renewcommand*{\arraystretch}{1.2}
\setlength{\tabcolsep}{2pt}
\newcommand{\tableentrybas}[2]{\colorbox{#1}{\parbox[c][0.9em][c]{122mm}{\centering\footnotesize #2}}}
\newcommand{\tableentrybass}[2]{\colorbox{#1}{\parbox[c][0.9em][c]{80mm}{\centering\footnotesize #2}}}
\newcommand{\tableentrybase}[2]{\colorbox{#1}{\parbox[c][2.1em][c]{18mm}{\centering\footnotesize #2}}}
\newcommand{\tableentrybasee}[2]{\colorbox{#1}{\parbox[c][0.9em][c]{38mm}{\centering\footnotesize #2}}}
\newcommand{\tableentrybaseeb}[2]{\colorbox{#1}{\parbox[c][2.6em][c]{38mm}{\centering\footnotesize #2}}}
\newcommand{\tableentrybaseee}[2]{\colorbox{#1}{\parbox[c][2.9em][c]{40mm}{\centering\footnotesize #2}}}
\newcommand{\tableentrybaseeee}[2]{\colorbox{#1}{\parbox[c][6em][c]{38mm}{\centering\footnotesize #2}}}
\footnotesize
\begin{tabular}{c|c|c}
& \multicolumn{2}{c}{Number of colors}\\ 
Graphs & $k=2$ & $k\geq 3$ \\
\cmidrule(lr){1-1}\cmidrule(lr){2-2}\cmidrule(lr){3-3}\\[-1.6em]

{Simple graphs} & 
\multicolumn{2}{c}{\tableentrybas{Red!20}{$1/n^{1-\varepsilon} $-inapprox.\ for $\varepsilon >0$ (Thm.~\ref{thm:maxpt-2})}}\\ 
\cmidrule(lr){1-1}\cmidrule(lr){2-2}\cmidrule(lr){3-3}\\[-1.7em]

{Multigraphs} & 
\multicolumn{2}{c}{\tableentrybas{Red!20}{$1/n^{1-\varepsilon}$-inapprox.\ for $\varepsilon >0$ (Thm.~\ref{thm:maxpt-2})}}\\ 
\cmidrule(lr){1-1}\cmidrule(lr){2-2}\cmidrule(lr){3-3}\\[-1.6em]

\multirow[c]{2}{*}{\tableentrybase{Red!0}{Complete \\graphs}} & 
\multirow[c]{2}{*}[0.1em]{\tableentrybaseeb{Green!30}{P \cite{bang1997alternating}}}&
\tableentrybass{Red!20}{MAX-SNP-hard (Thm.~\ref{thm:maxpt-3})}\\
& & \tableentrybass{Green!30}{$1/\sqrt{(2+\varepsilon )n}$-approx.\ for any $\varepsilon >0$  (Thm.~\ref{thm:maxptalg})}\\
\cmidrule(lr){1-1}\cmidrule(lr){2-2}\cmidrule(lr){3-3}\\[-1.6em]

\multirow[c]{2}{*}{\tableentrybase{Red!0}{Complete \\multigraphs}} & 
\multirow[c]{2}{*}[0.1em]{\tableentrybaseeb{Green!30}{P \cite{bang1997alternating}}}&
\tableentrybass{Red!20}{MAX-SNP-hard (Thm.~\ref{thm:maxpt-3})}\\
& & \tableentrybass{Green!30}{$1/\sqrt{(2+\varepsilon )n}$-approx.\ for any $\varepsilon >0$ (Thm.~\ref{thm:maxptalg})}
\end{tabular}
\caption{Complexity landscape of \maxpt.   \label{table:MaxPTresults}}
\end{table}

\vspace{-0.5cm}

%%%%%%%%%%%%%%%%%%%%%%%%%%%%%%%%
\paragraph{Paper Organization}
%%%%%%%%%%%%%%%%%%%%%%%%%%%%%%%%

The paper is organized as follows. In Section~\ref{sec:prelim}, we introduce basic definitions and notation, and overview some results of matroid theory that we will use in our proofs. In Section~\ref{sec:hardness}, we discuss the complexity of the \maxpf and \maxpt problems. The rest of the paper is devoted to presenting approximation algorithms mainly for \maxpf in various settings. In Section~\ref{sec:prep}, we show that the vertex set of the graph can be assumed to be coverable by monochromatic matchings of the graph, and that such a reduction can be found efficiently using techniques from matroid theory. We then give a polynomial algorithm for $2$-edge-colored complete multigraphs in Section~\ref{sec:complete}. Our main result is an $5/9$-approximation algorithm for the problem in $k$-edge-colored multigraphs, presented in Section~\ref{sec:general}. In Section~\ref{sec:simple}, we explain how the approximation factor can be improved if the graph is simple or the number of colors is at most three. We further improve the approximation factor for 2- and 3-edge-colored simple graphs in Section~\ref{sec:small}. Finally, an approximation algorithm is given for \maxpt in Section~\ref{sec:maxptapx}.

%%%%%%%%%%%%%%%%%%%%%%%%%%%%%%%%
\section{Preliminaries}
\label{sec:prelim}
%%%%%%%%%%%%%%%%%%%%%%%%%%%%%%%%

%%%%%%%%%%%%%%%%%%%%%%%%%%%%%%%%
\paragraph{Basic notation}
%%%%%%%%%%%%%%%%%%%%%%%%%%%%%%%%

We denote the set of \emph{nonnegative integers} by $\bZ_+$. For a positive integer $k$, we use $[k]\coloneqq\{1,\dots,k\}$. Given a ground set $S$, the \emph{difference} of $X,Y\subseteq S$ is denoted by $X\setminus Y$. If $Y$ consists of a single element $y$, then $X\setminus \{y\}$ and $X\cup \{y\}$ are abbreviated as $X-y$ and $X+y$, respectively.

We consider loopless undirected graphs possibly containing parallel edges. A graph is \emph{simple} if it has no parallel edges, and it is called a \emph{multigraph} if parallel edges might be present. A simple graph is \emph{complete} if it contains exactly one edge between any pair of vertices. By a \emph{complete multigraph}, we mean a multigraph containing at least one edge between any pair of vertices. A graph is \emph{linear} if each of its vertices has degree at most $2$ in it.

Let $G=(V,E)$ be a graph, $F\subseteq E$ be a subset of edges, and $X\subseteq V$ be a subset of vertices. The \emph{subgraph of $G$} and \emph{set of edges induced by $X$} are denoted by $G[X]$ and $E[X]$, respectively. The \emph{graph obtained by deleting $F$ and $X$} is denoted by $G-F-X$. We denote the \emph{vertices of the edges in $F$} by $V(F)$, and the \emph{vertex sets of the connected components of the subgraph $(V(F),F)$} by $\comp(F)\subseteq 2^{V(F)}$. We denote the \emph{set of edges in $F$ having exactly one endpoint in $X$} by $\delta_F(X)$ and define the \emph{degree of $X$ in $F$} as $d_F(X)\coloneqq |\delta_F(X)|$. We dismiss the subscript if $F=E$. A \emph{matching} is a subset of edges $M\subseteq E$ satisfying $d_M(v)\leq 1$ for every $v\in V$. We say that $F$ \emph{covers} $X$ if $d_F(v)\geq 1$ for every $v\in X$, or in other words, if $X\subseteq V(F)$.

Let $c\colon E\to[k]$ be an edge-coloring of $G$ using $k$ colors. The function $c$ is extended to subsets of edges where, for a subset $F\subseteq E$ of edges, $c(F)$ denotes the set of colors appearing on the edges of $F$. For an edge-colored graph $G=(V,E)$, we use $E_i=\{ e\in E\mid c(e)=i\}$ to denote the edges of color $i$. Without loss of generality, we assume throughout that $E_i$ contains no parallel edges. We call a subset of vertices $U\subseteq V$ \emph{matching-coverable} if there exist matchings $M_i\subseteq E_i$ for $i\in[k]$ such that $\bigcup_{i=1}^k M_i$ covers $U$. A \emph{properly colored 1-path-cycle factor} of a graph $G$ is a spanning subgraph consisting of a properly colored path $C_0$ and a (possibly empty) collection of properly colored cycles $C_1,\dots,C_q$ such that $V(C_i)\cap V(C_j)=\emptyset$ for $0\leq i <j\leq q$. We will use the following result of Bang-Jensen and Gutin~\cite{bang1997alternating}, extended by Feng, Giesen, Guo, Gutin, Jensen, and Rafiey~\cite{feng2006characterization}.

\begin{thm}[Bang-Jensen and Gutin~\cite{bang1997alternating}]\label{thm:bang}
A 2-edge-colored complete graph $G$ has a properly colored Hamiltonian path if and only if $G$ contains a properly colored 1-path-cycle factor. Furthermore, any properly colored 1-path-cycle factor of $G$ can be transformed into a properly colored Hamiltonian path in polynomial time. 
\end{thm}

For our approximation algorithm for \maxpt in complete graphs, we will rely on the following result of Borozan et al.~\cite{borozan2019maximum}.

\begin{thm}[Borozan et al.~\cite{borozan2019maximum}]\label{thm:maxpt-partition}
Let $G=(V,E)$ be an edge-colored complete multigraph. Then, there exists an efficiently computable partition $V_1\cup V_2$ of $V$ such that \maxpt can be solved in polynomial-time in both $G[V_1]$ and $G[V_2]$. Furthermore, the optimal solution $F_1$ in $G[V_1]$ is a properly colored spanning tree of $G[V_1]$.
\end{thm}

%%%%%%%%%%%%%%%%%%%%%%%%%%%%%%%%
\paragraph{Matroids}
%%%%%%%%%%%%%%%%%%%%%%%%%%%%%%%%

For basic definitions on matroids and on matroid optimization, we refer the reader to~\cite{oxley2011matroid,frank2011connections}. A \emph{matroid} $\cM=(E,\cI)$ is defined by its \emph{ground set} $E$ and its \emph{family of independent sets} $\cI\subseteq 2^E$ that satisfies the \emph{independence axioms}: (I1) $\emptyset\in\cI$, (I2) $X\subseteq Y,\ Y\in\cI\Rightarrow X\in\cI$, and (I3) $X,Y\in\cI,\ |X|<|Y|\Rightarrow\exists e\in Y\setminus X\ s.t.\ X+e\in\cI$. Members of $\cI$ are called \emph{independent}, while sets not in $\cI$ are called \emph{dependent}. The \emph{rank} $r_\cM(X)$ of a set $X$ is the maximum size of an independent set in $X$.

The \emph{union} or \emph{sum} of $k$ matroids $\cM_1=(E,\cI_1),\dots,\cM_k=(E,\cI_k)$ over the same ground set is the matroid $\cM_\Sigma=(E,\cI_\Sigma)$ where $\cI_\Sigma=\{I_1\cup\dots \cup I_k\mid I_i\in\cI_i\ \text{for each $i\in[k]$}\}$. Edmonds and Fulkerson~\cite{edmonds1965transversals} showed that the rank function of the sum is $r_{\cM_\Sigma}(Z)=\min\{\sum_{i=1}^k r_i(X)+|Z-X|\mid X\subseteq Z\}$, and provided an algorithm for finding a maximum sized independent set of $\cM_\Sigma$, together with its partitioning into independent sets of the matroids appearing in the sum, assuming an oracle access\footnote{In matroid algorithms, it is usually assumed that the matroid is given by a \emph{rank oracle} and the running time is measured by the number of oracle calls and other conventional elementary steps. For a matroid $M=(E,\cI)$ and set $X\subseteq E$ as an input, a rank oracle returns $r_M(X)$.} to the matroids $M_i$.

For an undirected graph $G=(V,E)$, the \emph{matching matroid} of $G$ is defined on the set of vertices $V$ with a set $X\subseteq V$ being independent if there exists a matching $M$ of $G$ such that $X\subseteq V(M)$, that is, $M$ covers all the vertices in $X$. Determining the rank function of the matching matroid is non-obvious since it requires the knowledge of the Berge-Tutte formula on the maximum cardinality of a matching in a graph. Nevertheless, the rank of a set can still be computed in polynomial time, see~\cite{edmonds1965transversals} for further details.

%%%%%%%%%%%%%%%%%%%%%%%%%%%%%%%%
\paragraph{MAX-SNP-hardness}
%%%%%%%%%%%%%%%%%%%%%%%%%%%%%%%%

While studying APX problems that are not in PTAS, Papadimitriou and Yannakakis~\cite{papadimitriou1991optimization} showed that a large subset of APX problems are in fact equivalent in this regard, meaning that either all of them belong to PTAS, or none of them do. By relying on the fundamental result of Fagin~\cite{fagin1974generalized} stating that existential second-order logic captures NP, they introduced the complexity class MAX-SNP that is contained within APX, together with a notion of approximation-preserving reductions, called L-reductions.  Given two optimization problems $A$ and $B$ with cost functions $c_A$ and $c_B$, respectively, a pair $f,g$ of polynomially computable functions is called an \emph{L-reduction} if there exists $\alpha,\beta>0$ such that (1) if $x$ is an instance of problem $A$ then $f(x)$ is an instance of problem $B$ and $\opt_B(f(x))\leq\alpha\cdot\opt_A(x)$, (2) if $y$ is a solution to $f(x)$ then $g(y)$ is a solution to $x$ and $|\opt_A(x)-c_A(g(y))|\leq\beta\cdot|\opt_B(f(x))-c_B(y)|$. This idea led to the definitions of MAX-SNP-complete and MAX-SNP-hard problems. In a seminal paper, Arora, Lund, Motwani, Sudan, and Szegedy~\cite{arora1998proof} proved that MAX-SNP-hard problems do not admit PTAS unless P$=$NP, hence one can think of MAX-SNP-complete problems as the class of problems having constant-factor approximation algorithms, but no approximation schemes unless P$=$NP. For example, Metric TSP, MAX-SAT, and Maximum Independent Set in Degree Bounded Graphs are prime examples of MAX-SNP-hard problems.

An instance of \tsp consists of a complete graph on $n$ vertices with all edge lengths being either 1 or 2. The \emph{length-1-degree} of a vertex is its degree in the subgraph of edges of length 1.
The current best inapproximability result for (1,2)-TSP is due to Karpinski and Schmied~\cite{karpinski2012approximation}, giving a constant lower bound on the approximability of the problem in general.

\begin{thm}[Karpinski and Schmied~\cite{karpinski2012approximation}]\label{thm:3534}
\tsp is NP-hard to approximate within a factor strictly smaller than $535/534$. 
\end{thm}

Theorem~\ref{thm:3534}, together with the result of Csaba, Karpinski and Krysta~\cite[Lemma 6.1]{csaba2002approximability} implies the following, stronger inapproximability bound. 

\begin{thm}[Csaba, Karpinski and Krysta\cite{csaba2002approximability}]\label{thm:Csaba}

For any $\varepsilon < 1/534$, there exists $0<d_0<1/2$ such that \tsp is NP-hard to approximate within a factor of $1+\varepsilon$ even for instances where the optimum is $n$ and the minimum length-1-degree is at least $d_0\cdot n$. 
\end{thm}

De la Vega and Karpinski~\cite{de1999approximation} proved MAX-SNP-hardness of the problem under similar assumptions.

\begin{thm}[De la Vega and Karpinski~\cite{de1999approximation}]\label{thm:MAX-SNP}
    For any $0<d_0<1/2$, \tsp is MAX-SNP-hard even for instances where the minimum length-1-degree is at least $d_0\cdot n$.
\end{thm}

In the \emph{Longest Path} problem (\longP), we are given a  directed graph $D=(V,A)$ on $n$ vertices and the goal is to find a directed path of maximum length in $D$. Björklund, Husfeldt and Khanna \cite{bjorklund2004approximating} showed the following. 

\begin{thm}[Björklund, Husfeldt and Khanna~\cite{bjorklund2004approximating}]\label{thm:longP}
\longP is NP-hard to approximate within a factor of $1/n^{1-\varepsilon}$ for any $\varepsilon >0$ even for instances containing a directed Hamiltonian path.
\end{thm}

For the undirected counterpart of the problem, called \emph{Undirected Longest Path} (\longUP), de la Vega and Karpinski~\cite{de1999approximation} proved the following result.

\begin{thm}[De la Vega and Karpinski~\cite{de1999approximation}]\label{longpsnp}
For any $0<d_0<\frac{1}{2}$, {\sc undirected}-\longP is MAX-SNP-hard even for instances where the minimum degree is at least $d_0\cdot n$.
\end{thm}

It is not difficult to see that this implies MAX-SNP-hardness of \longP\ too, even for instances where both the minimum in- and out-degree are at least $d_0\cdot n$.

%%%%%%%%%%%%%%%%%%%%%%%%%%%%%%%%
\section{Hardness results}
\label{sec:hardness}
%%%%%%%%%%%%%%%%%%%%%%%%%%%%%%%%

The aim of this section is to provide upper bounds on the approximability of \maxpf and \maxpt. We prove that \maxpf is MAX-SNP-hard for 2-edge-colored simple graphs as well as for 3-edge-colored simple complete graphs. Note that these imply analogous results for multigraphs and complete multigraphs, respectively. In the \emph{Maximum Linear Forest} problem (\maxlf), we are given an undirected graph $G=(V,E)$ and the goal is to find a linear forest of maximum size. In our proofs, we will rely on the following corollary of Theorem~\ref{thm:Csaba} and Theorem~\ref{thm:MAX-SNP}.

\begin{cor}\label{cor:12}
Let $0<\varepsilon < 1/534$ be an arbitrary constant. For any $0<d_0<\frac{1}{2}$, \maxlf\ is MAX-SNP-hard even for instances where the minimum degree is at least $d_0\cdot n$. Furthermore, there exists $0<d_0<1/2$ such that \maxlf is NP-hard to approximate within a factor of $1-\varepsilon$ even for simple Hamiltonian graphs with minimum degree at least $d_0\cdot n$.
\end{cor}
\begin{proof}
By Theorem \ref{thm:Csaba}, for any $0<\varepsilon<1/534$ there exists $0<d_0<1/2$ such that \tsp is NP-hard to approximate within a factor of $1+\varepsilon$ even for instances where the optimum is $n$, i.e., when the subgraph of length-1 edges is Hamiltonian, and the minimum length-1-degree is at least $d_0\cdot n$. Let $G$ be such an instance of \tsp. We construct an instance $G'$ of \maxlf by taking the subgraph of $G$ consisting of length-1 edges. Note that the minimum degree of $G'$ is exactly the minimum length-1-degree of $G$. Then, any linear forest containing at least $(1-\varepsilon)n$ edges for some $0<\varepsilon<1$ can be extended to a Hamiltonian cycle of length at most $(1-\varepsilon) n + 2\varepsilon n = (1+\varepsilon n)$ by adding length-2 edges connecting the endpoints of the components. Furthermore, any Hamiltonian cycle of length at most $(1+\varepsilon)n$ must contain at least $(1-\varepsilon)n$ length-1 edges, forming a linear forest in $G'$. Hence, for any $\varepsilon<1/534$, it is NP-hard to find a linear forest with at least $(1-\varepsilon)n$ edges, which shows the second statement. 

Let \opt denote the minimum length of tour in $G$ and $\opt'$ denote the maximum size of a linear forest in $G'$. By the above argument, from a linear forest $F'$ of size $x$ in $G'$ we can create in polynomial time a tour $F$ of length $x+2 (n-x)= 2n-x$ in $G$, which defines the function $g$. Vice versa, a tour $F$ of length $2n-x$ in $G$ implies a linear forest $F'$ of size $x$ in $G'$. Therefore, $\opt = 2n-\opt'$, so $\opt' \le 2n \le 2\cdot \opt$, since $\opt \ge n$. 
Finally we have $|\opt - (2n-x)| = |-\opt'+x| = |\opt '-x|$. Hence, we have an L-reduction with polynomially computable functions $f,g$ (where $f$ is the deletion of the length-2 edges from $G)$ and $\alpha =2$, $\beta = 1$.
This shows %that a PTAS for \maxlf for instances with minimum degree at least $d_0\cdot n$ implies a PTAS for (1,2)-TSP in such instances. By Theorem~\ref{thm:MAX-SNP}, this proves 
the MAX-SNP-hardness of the problem. 
\end{proof}

Using Theorem \ref{thm:3534}, an analogous argument gives the following.

\begin{cor}\label{cor:betterinappr}
\maxlf is NP-hard to approximate within a factor strictly smaller than $533/534$.
\end{cor}

%%%%%%%%%%%%%%%%%%%%%%%%%%%%%%%%
\subsection{Inapproximability of \maxpf}
\label{sec:hardnesspf}
%%%%%%%%%%%%%%%%%%%%%%%%%%%%%%%%

First we prove hardness of \maxpf in 2-edge-colored simple graphs.

\begin{thm}\label{thm:2snp}
For $2$-edge-colored simple graphs, \maxpf is MAX-SNP-hard. Furthermore, it is NP-hard to approximate within a factor strictly larger than $1601/1602$ even for instances containing a properly colored spanning tree. 
\end{thm}
\begin{proof}
We prove the statements by reduction from \maxlf. Consider an instance $G=(V,E)$ of \maxlf on $n$ vertices $\{v_1,\dots,v_n\}$. We construct an instance of \maxpf as follows. Let $G'$ and $G''$ be two copies of $G$, the edges of $G'$ being colored red and the edges of $G''$ being colored blue. For each vertex $v_i$ of $G$, let $v'_i$ be the copy of $v_i$ in $G'$ and $v''_i$ be the copy of $v_i$ in $G''$. For each $i\in[n]$, we add a vertex $u_i$ together with two new edges $v'_iu_i$ and $u_iv''_i$ having colors blue and red, respectively; see Figure~\ref{fig:reduc_from_tsp} for an example. The construction is polynomial and gives the function $f$.

We denote by $\hat{G}$ the graph thus obtained. Let \opt denote the maximum size of a linear forest in $G$ and $\opt'$ denote the maximum size of a properly colored forest in $\hat{G}$. We claim that $\opt'=\opt+2n$. Let $F$ be a linear forest in $G$. We create a properly colored forest $\hat{F}$ in $\hat{G}$ of size $|F|+2n$ as follows. First, we take a proper coloring of the edges of $F$ using colors red and blue. Note that such a coloring exists as $F$ is linear. Then, for each red edge $v_iv_j$ we add $v_i'v_j'$ to $\hat{F}$, and for each blue edge $v_iv_j$ we add $v_i''v_j''$ to $\hat{F}$. Finally, we add all the edges in $\{ u_iv_i',u_iv_i''\mid i\in [n]\}$ to $\hat{F}$. By the construction, we have $|\hat{F}|=|F|+2n$. Since each vertex had at most one red and one blue edge incident to it after coloring the edges of $F$, $\hat{F}$ is properly colored. Finally, $\hat{F}$ is a forest, as otherwise contracting the edges of the form $u_iv_i',u_iv_i''$ of a cycle $\hat{C}$ in $\hat{F}$ would result in a cycle $C$ in $F$, a contradiction. This implies $\opt'\geq \opt+2n$.

For the other direction, let $\hat{F}$ be a properly colored forest of size in $G'$. First, we create a properly colored forest $\hat{F}'$ such that $|\hat{F}'|\ge |\hat{F}|$ and $E_u=\{ u_iv_i',u_iv_i''\mid i\in [n]\} \subseteq \hat{F}'$. This is achieved by adding the edges of $E_u$ one by one. By the construction, whenever an edge $u_iv_i'$ is added to any properly colored forest in $\hat{G}$ then the forest does not contain any adjacent edges having the same color. Therefore, in order to maintain a properly colored forest, it suffices to delete at most one edge from a cycle that $u_iv_i'$ possibly creates, and the size of the forest does not decrease. Furthermore, if $u_iv_i'$ creates a cycle, then there must be another edge in the cycle incident to $v_i'$ which can be deleted, hence we never have to delete an edge in $E_u$ throughout. By similar arguments, edges of the type $u_iv_i''$ can also be added to the solution. Clearly, this transformation can be performed in polynomial time for any properly colored forest of $\hat{G}$. Therefore, assume that $\hat{F}$ is a properly colored forest such that $E_u\subseteq \hat{F}$. Then, contracting the edges in $E_u$ results in a forest $F$. Furthermore, $F$ is linear since each $v_i'$ and $v_i''$ had at most one incident not in $E_u$. That is, $F$ is a linear forest in $G$ of size $|F|=|\hat{F}|-2n$. This implies $\opt\geq\opt'-2n$.

By Corollary \ref{cor:12}, \maxlf is MAX-SNP-hard even if the minimum degree is at least $\frac{n}{3}$, hence we may assume that $\opt\ge 1/3\cdot n$.
We conclude that $\opt'=\opt+2n \le 7\cdot \opt$.  Furthermore, by the above argument, if we can find a properly colored forest $\hat{F}$ of size $x+2n$ in $\hat{G}$, then we can create a linear forest $F$ of size $x$ in $G$ in polynomial time, which defines the function $g$. Finally, we have that $| |F| -\opt | = | |\hat{F}|-2n-\opt|=||\hat{F}|-\opt'|$. Hence, we have constructed an L-reduction with $\alpha = 7$, $\beta = 1$, proving MAX-SNP-hardness.

For the second half, assume that $\opt =n-1$ and hence $\opt'=3n-1$, that is, $\hat{G}$ contains a properly colored spanning tree. By Corollary~\ref{cor:12}, \maxlf is NP-hard to approximate in such instances. Therefore, if there exists a $(1-\varepsilon )$-approximation algorithm for \maxpf for 2-edge-colored simple graphs containing a properly colored spanning tree, then it gives a properly colored forest of size at least $(1-\varepsilon )(3n-1)$ in $\hat{G}$. Using the argument above, this implies a linear forest in $G$ of size at least $(1-3\varepsilon )(n-1) -2\varepsilon$, and thus gives a $(1-3\varepsilon')$-approximation algorithm for \maxlf for any $\varepsilon'>\varepsilon$. By Corollary \ref{cor:12}, for any $\varepsilon <1/1602$, it is NP-hard to approximate \maxpf withing a factor of $(1-\varepsilon)$ even in 2-edge-colored simple graphs containing a properly colored spanning tree.
\end{proof}

\begin{figure}[t!]
\centering
\begin{subfigure}[t]{0.48\textwidth}
    \centering
    \includegraphics[width=0.5\textwidth]{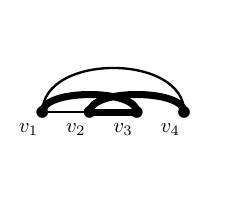}
    \caption{An instance of \maxlf. Thick edged form a maximum linear forest $F$.}
    \label{fig:33a}
\end{subfigure}\hfill
\begin{subfigure}[t]{0.48\textwidth}
    \centering
    \includegraphics[width=0.5\textwidth]{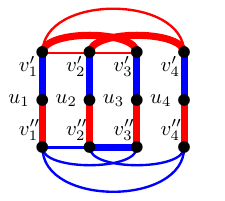}
    \caption{The corresponding instance of \maxpf in a 2-edge-colored simple graph. Thick edges form a maximum properly colored forest $\hat{F}$.}
    \label{fig:33b}
\end{subfigure}\hfill
\caption{Illustration of the proof of Theorem~\ref{thm:2snp}.}
\label{fig:reduc_from_tsp}
\end{figure}

For 3-edge-colored complete simple graphs, we get a slightly worse upper bound on the approximability of the problem.

\begin{thm}\label{thm:3snp}
For $3$-edge-colored complete, simple graphs, \maxpf is MAX-SNP-hard. Furthermore, it is NP-hard to approximate within a factor strictly larger than $3203/3204$ even for instances containing a properly colored spanning tree. 
\end{thm}
\begin{proof}
By the proof of Theorem \ref{thm:2snp}, \maxpf is MAX-SNP-hard even for 2-edge-colored simple graphs admitting a solution of size at least $n/3$. Let $G$ be such an instance of \maxpf on $n$ vertices $\{ v_1,\dots,v_n\}$. We create another instance $G'$ of \maxpf as follows. First, we take a copy of $G$ and keep the color of the edges. Then, we add $n$ new vertices $\{ v_{n+1},\dots,v_{2n}\}$ to $G$. Finally, we make the resulting graph complete on $2n$ vertices by adding an edge $v_iv_j$ with a third color whenever $v_iv_j\notin E$ for $i,j \in [2n],i\neq j$; see Figure~\ref{fig:34} for an example. This defines the function $f$, which is clearly polynomial time computable.

Let \opt and $\opt'$ denote the maximum size of a properly colored spanning tree in $G$ and $G'$, respectively. We claim that $\opt'=\opt+n$. Let $F$ be a properly colored forest in $G$. Then adding the edges $v_iv_{n+i}$ for $i\in [n]$ results in a properly colored forest $F'$ of $G'$ with $|F'|=|F|+n$. For the other direction, take a properly colored forest $F'$ of $G'$. Note that $F'$ contains at most $n$ edges not in $E$ since those have the same color and hence necessarily form a matching. We define $g$ by deleting these edges from $F'$, which results in a properly colored forest $F$ in $G$ with $|F|\ge|F'|-n$. 

Then, we have $\opt ' \le 4 \cdot \opt$, since by the proof of Theorem~\ref{thm:2snp}, we may assume $\opt \ge \frac{n}{3}$. Finally, we have $| |F| -\opt |= | |F'|-n - (\opt ' -n)|=| |F'|-\opt'|$. Hence, we have an L-reduction with $\alpha=4$, $\beta =1$, proving MAX-SNP-hardness.

For the second half, assume further that the instance $G$ that we reduce from contains a properly colored spanning tree. Recall that, by Theorem \ref{thm:2snp}, \maxpf is NP-hard to approximate even for such instance. Then we have $\opt' =\opt +n =2n-1$. Suppose that we have an $(1-\varepsilon )$-approximation algorithm for \maxpf in 3-edge-colored complete simple graphs containing a properly colored spanning tree. Then we can find a properly colored forest in $G'$ of size at least $(1-\varepsilon )(2n-1)$. By deleting the edges of the forest not in $E$, we get a properly colored forest in $G$ of size at least $(1-2\varepsilon )(n-1)-\varepsilon$. Hence, for any $\varepsilon '>\varepsilon$, an $(1-\varepsilon )$-approximation algorithm for \maxpf in 3-edge-colored complete simple graphs containing a properly colored spanning tree implies an $(1-2\varepsilon')$-approximation algorithm for \maxpf in 2-edge-colored simple graphs containing a properly colored spanning tree. By Theorem~\ref{thm:2snp}, for any $\varepsilon<1/3204$, it is NP-hard to approximate \maxpf withing a factor of $(1-\varepsilon)$ even in 3-edge-colored complete simple graphs containing a properly colored spanning tree.
\end{proof}

We also show a constant upper bound for the approximability of \maxpf in 2-edge-colored multigraphs. 

\begin{thm}
For $2$-edge-colored multigraphs, \maxpf is NP-hard to approximate within a factor strictly larger than $533/534$.
\end{thm}
\begin{proof}
The proof is by reduction from \tsp. Consider an instance of \tsp, that is, a complete simple graph $G$ on $n$ vertices with all edge lengths being either 1 or 2. We construct an instance of \maxpf as follows. Take the subgraph of edges of length 1, and replace each of its edges by two parallel copies, one being colored red and the other being colored blue. Let $G'$ denote the 2-edge-colored multigraph thus obtained. For ease of discussion, we denote by $\opt$ the minimum length of a Hamiltonian cycle in $G$ and by $\opt'$ the maximum size of a properly colored forest in $G'$. Clearly, $\opt\geq\opt'$.

Assume for a contradiction that \maxpf has a strictly better than $533/534$-approximation algorithm for 2-edge-colored multigraphs, and let $F'$ the output of the algorithm when applied to $G'$. Since $G'$ is a 2-edge-colored graph, $F'$ is a linear forest. The original copies of the edges appearing in $F'$ form a linear forest of $G$ that can be extended to a Hamiltonian cycle $C$ of total length $2n - |F'|$ by adding $n -|F'|$ edges of length 2 to it. An analogous argument shows that $\opt \ge 2n - \opt'$. Therefore, the total length of $C$ can be bounded as $2n-|F'| < 2n-533/534\cdot\opt' = 2n -\opt' + 1/534\cdot\opt' \le \opt + 1/534\cdot \opt' \le 535/534\cdot\opt$, contradicting Theorem~\ref{thm:3534}.
\end{proof}

\begin{figure}[t!]
\centering
\begin{subfigure}[t]{0.48\textwidth}
    \centering
    \includegraphics[width=0.5\textwidth]{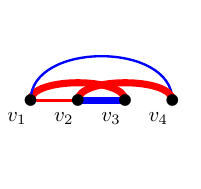}
    \caption{An instance of \maxpf in a 2-edge-colored simple graph. Thick edged form a maximum properly colored forest $F$.}
    \label{fig:34a}
\end{subfigure}\hfill
\begin{subfigure}[t]{0.48\textwidth}
    \centering
    \includegraphics[width=0.5\textwidth]{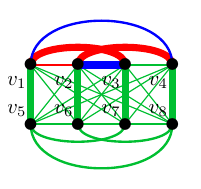}
    \caption{The corresponding instance of \maxpf in a 3-edge-colored complete simple graph. Thick edges form a maximum properly colored forest $F'$.}
    \label{fig:34b}
\end{subfigure}\hfill
\caption{Illustration of the proof of Theorem~\ref{thm:3snp}.}
\label{fig:34}
\end{figure}

%%%%%%%%%%%%%%%%
\subsection{Inapproximability of \maxpt}
%%%%%%%%%%%%%%%%

In general, \maxpt cannot be approximated within a polynomial factor even if the graph assumed to contain a properly colored spanning tree.

\begin{thm}\label{thm:maxpt-2} 
For 2-edge-colored simple graphs, \maxpt is MAX-SNP-hard. Furthermore, it is NP-hard to approximate within a factor of $1/n^{1-\varepsilon}$ for any $\varepsilon >0$ even for instances containing a properly colored spanning tree.
\end{thm}
\begin{proof}
We prove by reduction from \longP, which does not admit a $1/n^{1-\varepsilon}$-approximation algorithm for any $\varepsilon >0$, unless $P=NP$; see Theorem \ref{thm:longP}. Let $D=(V,A)$ be an instance of \longP on $n$ vertices $\{v_1,\dots,v_n\}$. We create an instance $G'=(V',E')$ of \maxpt with colors red and blue as follows. For each vertex $v_i\in V$, we add $2$ vertices $\{v_i^{in},v_i^{out}\}$ to $V'$. For each $v_iv_j\in A$, we add a blue edge $v_i^{out}v_j^{in}$ to $E'$. Finally, for every vertex $v_i\in V$, we add a red edge $v_i^{in}v_i^{out}$ to $E'$. Let $n'=|V'|=2n$. This defines the function $f$, which is clearly polynomial. Let \opt denote the maximum length of a directed path in $D$, and let $\opt'$ denote the maximum size of a properly colored tree in $G$. Since $G'$ is colored using two colors, a properly colored tree is a path with alternating edge colors. 

Let $P$ be a longest directed path in $D$, and let $p-1$ denote the number of its edges. First we show that $G'$ admits a path of length $2p-1$ with alternating edge colors. Indeed, for the path $P=\{ v_{i_1}v_{i_2}\dots v_{i_p}\}$ in $D$, we can create an alternating path of length $2p-1$ by taking $P'=\{ v_{i_1}^{in}v_{i_1}^{out}v_{i_2}^{in}v_{i_2}^{out}\dots v_{i_p}^{in}v_{i_p}^{out}\}$. This implies $\opt'\geq 2\cdot \opt+1$.

For the other direction, given an alternating path $P'$ in $G'$ of length $p'-1$, then we can create a directed path $P$ in $D$ of length at least $\lceil \frac{p'-2}{2} \rceil$. To see this, take any alternating path $P'$ in $G$ and let $v_i^{in}v_i^{out}\in P'$. Then $v_i^{in}v_i^{out}$ is either followed by a copy $v_i^{out}v_j^{in}$ of an original edge $v_iv_j$, or it is the last edge of the path. Similarly, $v_i^{in}v_i^{out}$ is either preceded by an edge $v_j^{out}v_i^{in}$, or it is the first edge of the path. If at least one of the first and last edges of the path is of the form $v_i^{in}v_i^{out}$, then it follows that the two endpoints are copies of different vertices. Since there is at most one edge of the form $v_j^{out}v_i^{in}$ ($i\ne j$) incident to $v_i^{in}$ and at most one of the form $v_i^{out}v_j^{in}$ ($i\ne j$) incident to $v_i^{out}$, it follows that contracting the edges of the form $v_i^{in}v_i^{out}$ of $P'$ results either a directed path or a directed cycle in $D$. By the above, the resulting subgraph has at least $\frac{p'}{2}$ edges and is a directed path or a cycle, or it has at least $\frac{p'-2}{2}$ edges and it is a directed path. In both cases, we get a directed path $P$ having length at least $\lceil \frac{p'-2}{2}\rceil$. This defines the function $g$. Therefore, we get that $2\cdot \opt+1\ge \opt'$. 

We conclude that $\opt' =2\cdot \opt +1\le 3\cdot \opt$. Also, $ ||P|-\opt | \le |(|P'|-1)/2-\opt |= |(|P'|-1)/2-(\opt '-1)/2| = 1/2\cdot ||P'|-\opt'|$. Hence, we have an L-reduction with $\alpha =3$, $\beta = 1/2$, proving MAX-SNP-hardness by Theorem \ref{longpsnp}.

For the second half of the theorem, observe that, a $1/n'^{1-\varepsilon}$-approximation algorithm for \maxpf in $G'$ implies that we can find a directed path in $D$ of length at least $1/n'^{1-\varepsilon}\cdot (\opt'-1)/2= 1/(2n)^{1-\varepsilon}\cdot \opt$ for any $\varepsilon >0$, and so implies an $n^{1-\varepsilon '}$-approximation algorithm for \longP for some $0<\varepsilon '<\varepsilon$. By Theorem~\ref{thm:longP}, for any $\varepsilon>0$, it is NP-hard to approximate \maxpt withing a factor of $1/n^{1-\varepsilon}$ even for instances containing a properly colored spanning tree. 
\end{proof}

\begin{figure}[t!]
\centering
\begin{subfigure}[t]{0.48\textwidth}
    \centering
    \includegraphics[width=0.5\textwidth]{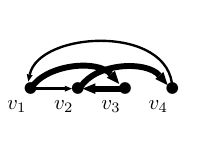}
    \caption{An instance of \longP. Thick edged form a longest directed path $P$.}
    \label{fig:36a}
\end{subfigure}\hfill
\begin{subfigure}[t]{0.48\textwidth}
    \centering
    \includegraphics[width=0.5\textwidth]{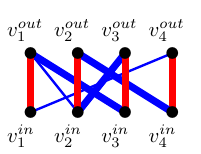}
    \caption{The corresponding instance of \maxpt in a 2-edge-colored simple graph. Thick edges form a maximum alternating path $P'$.}
    \label{fig:36b}
\end{subfigure}\hfill
\caption{Illustration of the proof of Theorem~\ref{thm:maxpt-2}.}
\label{fig:36}
\end{figure}

Finally, we prove an inapproximability result that is independent from the assumption P$\neq$NP.

\begin{thm}\label{thm:maxpt-3}
For $3$-edge-colored complete simple graphs, \maxpt\ is MAX-SNP-hard. Furthermore, it is NP-hard to approximate within a factor strictly larger than $3203/3204$ even for instances containing a properly colored spanning tree.
\end{thm}
\begin{proof}

We reduce from \maxlf. Let $G$ be an instance of \maxlf\ with $n$ vertices. We use the same construction as in Theorem \ref{thm:2snp} to create a simple 2-edge-colored graph $G'$ on $3n$ vertices. Then, we use the construction from Theorem \ref{thm:3snp} to create a complete, simple 3-edge colored graph $G''$ on $6n$ vertices. This defines the function $f$.

Suppose we have an inclusionwise maximal linear forest $F$ in $G$. Then, we have that $F$ is a union of paths, such that there are no edges between the endpoints of the paths. As in Theorem \ref{thm:2snp}, we can create a properly colored forest $F'$ of size $|F|+2n$ in $G'$. By the above observation, we have that in $G''$ we can create a properly colored tree that covers all $3n$ vertices of $G'$, by ordering the paths $P_1,\dots, P_k$ in $F'$ and adding edges $v_{i_2}''v_{(i+1)_1}'$, where $v_{i_2}$ is the last vertex of the path $P_i$ and $v_{(i+1)_1}$ is the first vertex of the path $P_{i+1}$ for $i=1,\dots ,k-1$, because all these edges are of the new color green. Finally, we have that there is $2\cdot (3n-|F'|-1)=2n-2|F|-2$ vertices of $G'$ that now have an adjacent edge of color green. Hence, we can add $n+2|F|+2$ more edges of color green that go to the other $3n$ vertices to the properly colored spanning tree in $G'$. Hence, we can create a properly colored tree $F''$ of size $2|F|+4n+1$ in $G''$. Therefore, $\opt'' \ge 2\opt + 4n +1$.

For the other direction, suppose we have a properly colored forest $F''$ in $G''$. %Let $\ell =\ell ' +\ell ''$ denote the number of green edges of $F''$, where $\ell'$ is the number of green edges induced by $G'$ and $\ell ''$ is the number of other green edges. It is easy to see that each green edge has at least one endpoint in $G'$, because otherwise there can be no adjacent edge to it in a properly colored forest, as the vertices in $G''\setminus G'$ have only green adjacent edges. This implies that $2\ell'+\ell''\le 3n$.
As $F''$ is a tree, we also have that the number of green edges can be at most $(|F''|+1)/2$. Hence, the number of edges from the other two colors is at least $(|F''|-1)/2$. Therefore, we can get a properly colored forest $F'$ of size $(|F''|-1)/2$ in $G'$ by deleting the green edges. Similarly as in Theorem \ref{thm:2snp}, we can get a linear forest of size at least  $(|F''|-1)/2 -2n$ in $G$, which defines the function $g$.
By Corollary \ref{cor:12} we have that MAX-SNP-hardness of \maxlf\ remains even if $\opt \ge \frac{n}{3}$.
Hence, $\opt '' \le 15\cdot \opt$. Also, $||F|-\opt | \le | (|F''|-1)/2-2n -((\opt ''-1)/2-2n) |\le  1/2\cdot ||F''|-\opt ''|$.
Therefore, we have an L-reduction with $\alpha = 15$, $\beta = 1/2$, which proves MAX-SNP-hardness.

By Theorem \ref{thm:3snp}, \maxpf is NP-hard to approximate within a factor of $1-\varepsilon$ for any $\varepsilon <1/3204$, even in 3-edge-colored complete simple graphs containing a properly colored spanning tree. Furthermore, observe that for such graphs, a $(1-\varepsilon )$-approximation algorithm for \maxpt gives a properly colored tree of size at least $(1-\varepsilon)(n-1)$, hence it gives a $(1-\varepsilon )$-approximation for \maxpf as well, concluding the proof.
\end{proof}
%%%%%%%%%%%%%%%%%%%%%%%%%%%%%%%%

%%%%%%%%%%%%%%%%%%%%%%%%%%%%%%%%
\section{Approximation algorithms}
\label{sec:approx}
%%%%%%%%%%%%%%%%%%%%%%%%%%%%%%%%

In this section, we provide approximation algorithms for \maxpf and \maxpt in various settings. First, in Section~\ref{sec:prep}, we establish a connection between \maxpf and the sum of matching matroids defined by the color classes of the coloring of the graph. In Section~\ref{sec:complete}, we discuss 2-edge-colored complete multigraphs and show that \maxpf is solvable in polynomial time for this class. Our main result is a general $5/9$-approximation algorithm for \maxpf in multigraphs, presented in Section~\ref{sec:general}. In Section~\ref{sec:simple}, we explain how the approximation factor can be improved if the graph is simple or the number of colors is at most three, and then we further improve the approximation factor for 2- and 3-edge-colored simple graphs in Section~\ref{sec:small}. Finally, an approximation algorithm for \maxpt is given in Section~\ref{sec:maxptapx}. 

Throughout the section, we denote by $\opt[G]$ the size of an optimal solution for the underlying problem, i.e., \maxpf or \maxpt, in graph $G$.

%%%%%%%%%%%%%%%%%%%%%%%%%%%%%%%%
\subsection{Preparations}
\label{sec:prep}
%%%%%%%%%%%%%%%%%%%%%%%%%%%%%%%%

For analyzing the proposed algorithms, we need some preliminary observations. Consider an instance of \maxpf, that is, a $k$-edge-colored graph $G=(V,E)$ on $n$ vertices. Recall that $E_i$ denotes the set of edges colored by $i$ and that a subset of vertices $U\subseteq V$ is called matching-coverable if there exist matchings $M_i\subseteq E_i$ for $i\in[k]$ such that $\bigcup_{i=1}^k M_i$ covers $U$. Using matroid terminology, this is equivalent to $U$ being independent in the sum of the matching matroids defined by the color classes. The next lemma shows that it suffices to restrict the problem to a maximum sized matching-coverable set.

\begin{lem}
\label{lemma:key}
For any matching-coverable set $U\subseteq V$, there exists a maximum-size properly colored forest $F_{opt}$ in $G$ such that $d_{F_{opt}}(u)\geq 1$ for every $u\in U$. Furthermore, if $U$ is a maximum-size matching-coverable set, then $\opt[G] = \opt[G[U]]$. 
\end{lem}
\begin{proof}
Let $U\subseteq V$ be a matching-coverable set and let $M_1,\dots,M_k$ be matchings satisfying $M_i\subseteq E_i$ and $U\subseteq V(\bigcup_{i=1}^k M_i)$. Let $F_{opt}$ be a maximum-size properly colored forest in $G$ that has as many edges in common with $M_1\cup \dots \cup M_k$ as possible. We claim that $F_{opt}$ covers $U$. Suppose indirectly that there exists a vertex $u\in U$ that is not covered by $F_{opt}$. For any edge $e\in M_1\cup \dots \cup M_k$ incident to $u$, $F_{opt}+e$ is still a forest by the indirect assumption. Moreover, $F_{opt}$ contains at most one edge adjacent to $e$ having the same color as $e$. Since $F_{opt}$ has maximum size, there exists exactly one such edge $f$. However, as $f$ has the same color as $e\in M_1\cup \dots \cup M_k$, we get that $f\notin M_1\cup \dots \cup M_k$. Therefore, $F_{opt}-f+e$ is a maximum-size properly colored forest containing more elements from $M_1\cup\dots\cup M_k$ than $F_{opt}$, a contradiction. This proves the first half of the lemma. 

To see the second half, let $U$ be a maximum-size matching-coverable set and $F_{opt}$ be a maximum-size properly colored forest covering $U$, implying $U\subseteq V(F_{opt})$. Note that $N_i=E_i\cap F_{opt}$ is a matching for every $i\in[k]$, hence $V(F_{opt})$ is also a matching-coverable set. By the maximality of $U$, we get $U=V(F_{opt})$, concluding the proof.
\end{proof}

\begin{rem}\label{rem:alg}
Since a rank oracle for the matching matroid of a graph can be constructed in polynomial time~\cite{edmonds1965transversals}, a maximum-size matching-coverable set $U$ can be found by using the matroid sum algorithm of Edmonds and Fulkerson~\cite{edmonds1965transversals}. The algorithm also provides a partition $U=U_1\cup \cdots \cup U_k$ where $U_i$ is independent in the matching matroid defined by $E_i$. For each $U_i$, one can find a matching $M_i\subseteq E_i$ that covers $U_i$ using Edmonds' matching algorithm~\cite{edmonds1965paths}. Furthermore, each matching $M_i$ can be chosen to be a maximum matching in $E_i$, due the the underlying matroid structure. Concluding the above, a maximum-size matching-coverable set $U$ together with maximum matchings $M_1,\dots, M_k$ with $M_i\subseteq E_i$ and $V(\bigcup_{i=1}^k M_i)=U$ can be found in polynomial time.    
\end{rem}

%%%%%%%%%%%%%%%%%%%%%%%%%%%%%%%%
\subsection{2-edge-colored complete multigraphs}
\label{sec:complete}
%%%%%%%%%%%%%%%%%%%%%%%%%%%%%%%%

Though \maxpf is hard even to approximate in general, the problem turns out to be tractable for $2$-edge-colored complete multigraphs. Our algorithm is presented as Algorithm \ref{Polyalgo:maxpf}.

\begin{algorithm}[t!]
\caption{Algorithm for \maxpf in 2-edge-colored complete multigraphs.}\label{Polyalgo:maxpf}
\DontPrintSemicolon

\KwIn{A 2-edge-colored complete multigraph $G=(V,E)$.}
\KwOut{A properly colored forest $F$.}

\medskip

Find maximum matchings $M_1\subseteq E_1, M_2\subseteq E_2$ maximizing $|V(M_1\cup M_2)|$.\;
Let $F \coloneqq M_1 \cup M_2$ and $U\coloneqq V(F)$.\;
Let $\mathcal{P}$ and $\mathcal{C}$ denote the path and cycle components in $\comp(F)$, respectively.\;
\If{$\mathcal{P}=\emptyset$}{
\begin{varwidth}{0.85\linewidth}
Delete any edge of $F$, transform the remaining set of edges into a properly colored Hamiltonian path $P$ using Theorem~\ref{thm:bang}, and update $F\gets P$.\label{Polystep:cycle}\;    
\end{varwidth}
}
\Else{
\begin{varwidth}{0.85\linewidth}
Let $P\in\mathcal{P}$ arbitrary and let $F'\coloneqq F[P\cup\bigcup_{C\in\mathcal{C}} C]$.\;
Transform $F'$ into a properly colored Hamiltonian path $P'$ using Theorem~\ref{thm:bang} and update $F\gets (F\setminus F')\cup P'$.\label{Polystep:path}\;
\end{varwidth}
}
\Return F\;
\end{algorithm}

\begin{thm}\label{thm:ck2}
Algorithm \ref{Polyalgo:maxpf} outputs a maximum-size properly colored forest for $2$-edge-colored complete multigraphs in polynomial time.
\end{thm}
\begin{proof}
Note that each component of $M_1\cup M_2$ is either a path or a cycle whose edges are alternating between $M_1$ and $M_2$. If $M_1 \cup M_2$ is the union of cycles, then Algorithm~\ref{Polyalgo:maxpf} gives a properly colored Hamiltonian path in $G[U]$ by Step \ref{Polystep:cycle}. By Lemma~\ref{lemma:key}, $G[U]$ contains a maximum-size properly colored forest and hence $\opt\leq |U|-1$, implying that $F$ is optimal.

If $M_1 \cup M_2$ has a path component, then Step~\ref{Polystep:path} of Algorithm~\ref{Polyalgo:maxpf} does not reduce the number of edges, i.e., the output $F$ of the algorithm has size $|M_1|+|M_2|$. Since $M_1$ and $M_2$ were chosen to be maximum matchings in $E_1$ and $E_2$, respectively, the sum of their sizes is clearly an upper bound on the maximum size of a properly colored forest, implying that $F$ is optimal.

The overall running time of the algorithm is polynomial by Theorem~\ref{thm:bang} and Remark~\ref{rem:alg}.
\end{proof}

%%%%%%%%%%%%%%%%%%%%%%%%%%%%%%%%
\subsection{General case}
\label{sec:general}
%%%%%%%%%%%%%%%%%%%%%%%%%%%%%%%%

This section is dedicated for the proof of our main result, a general approximation algorithm for \maxpf. The high-level idea of our approach is as follows. With the help of Lemma~\ref{lemma:key}, we restrict the problem to a subgraph $G[U]$ where $U$ is a maximum-size matching-coverable set. Throughout the algorithm, we maintain maximum matchings $M_i\subseteq E_i$ for $i\in[k]$ such that $F=\bigcup_{i=1}^k M_i$ covers $U$. We then try to improve the structure of $F$ by decreasing the number of its components of size $2$ by local changes. These local improvement steps consist of adding one or two appropriately chosen edges. If no improvement is found, a careful analysis of the structure of the current solution gives a better-than-1/2 guarantee for the approximation factor.

Before stating the algorithm and the theorem, let us remark that there are several ways of getting a $1/2$-approximation for \maxpf in general. As it was mentioned already in Section~\ref{sec:intro}, the algorithms of~\cite{kiraly2012degree} and~\cite{linhares2020approximate} provide such a solution. However, there is a simple direct approach as well: find matchings $M_i \subseteq E_i$ for $i \in [k]$ maximizing the size of $U \coloneqq V(\bigcup_{i=1}^k M_i)$, and take a maximum forest $F$ in $\bigcup_{i=1}^k M_i$. This provides a $1/2$-approximation by Lemma~\ref{lemma:key}, since $|F| \ge |U|/2 \ge \opt[G[U]]/2 \ge \opt[G]$ holds.
However, improving the $1/2$ approximation factor is non-trivial and requires new ideas. Our main contribution is to break the 1/2 barrier and show that the problem can be approximated within a factor strictly better than 1/2. The algorithm is presented as Algorithm~\ref{algo:maxpf}.

\begin{algorithm}[t!] 
\caption{Approximation algorithm for \maxpf in multigraphs.}\label{algo:maxpf}
\DontPrintSemicolon

\KwIn{A multigraph $G=(V,E)$ with edge-coloring $c\colon E\to [k]$.}
\KwOut{A properly colored forest $F$ in $G$.}

\medskip

%Let $E_i\coloneqq\{e\in E\mid c(e)=i\}$ for $i\in[k]$. 
Find matchings $M_i\subseteq E_i$ for $i\in[k]$ maximizing $|\bigcup_{i=1}^k V(M_i)|$.{\footnotesize\color{Blue}{\tcp*[r]{Preprocessing steps.}}}
Let $F\coloneqq \bigcup_{i=1}^k M_i$ and $U\coloneqq V(F)$.\;
\SetAlgoVlined
$U_s\coloneqq\bigcup\{C\in\comp(F)\mid |C|=2\}$.\label{step:spr}     
{\footnotesize\color{Blue}{\tcp*[r]{Union of size-two components.}}}
$U_r\coloneqq U\setminus U_s$.
{\footnotesize\color{Blue}{\tcp*[r]{Remaining vertices.}}}
Take a maximum forest $F^\circ$ in $F[U_r]$ and set $F\gets (F\setminus F[U_r])\cup F^\circ$.\label{step:f'}{\footnotesize\color{Blue}{\tcp*[r]{Maximum forest in $U_r$.}}}
Let $E'\coloneqq E[U_s]\cup \{vw\in E\mid v\in U_s,w\in U_r,c(vw)\notin c(\delta_F(w))\}$.\label{step:e'}{\footnotesize\color{Blue}{\tcp*[r]{Candidate edges for extending $F^\circ$.}}}
Let $E'_i\coloneqq E' \cap E_i$.\; 
\For({{\footnotesize\color{Blue}{\tcp*[f]{Trying to add single edges.}}}}){$uv \in E\setminus F$ \textbf{with} $c(uv) \not \in c(\delta_F(u)\cup \delta_F(v))$ \label{step:for3}}{
\begin{varwidth}{0.9\linewidth}
If $u$ and $v$ are in different components of $F$, then $F \gets F+uv$ and go to Step~\ref{step:spr}.
\label{step:max}\;        
\end{varwidth}
}
\For({{\footnotesize\color{Blue}{\tcp*[f]{Trying to improve using single edges.}}}}){$uv\in E'$ \textbf{with} $u\in U_s,v\in U_r$\label{step:for1}}{
\begin{varwidth}{0.9\linewidth}
If there exist matchings $N_i\subseteq E'_i$ for $i\in[k]$ such that $uv\in N_{c(uv)}$ and $U_s+v\subseteq V(\bigcup_{i=1}^k N_i)$, then $F\gets (F\setminus F[U_s])\cup(\bigcup_{i=1}^k N_i)$ and go to Step~\ref{step:spr}.\label{step:case1}\;
\end{varwidth}        
}
\For({{\footnotesize\color{Blue}{\tcp*[f]{Trying to improve using pairs of edges.}}}}){$uv_1,uv_2\in E'[U_s]$ \textbf{with} $v_1\neq v_2$, $c(uv_1)\neq c(uv_2)$\label{step:for2}}{
\begin{varwidth}{0.9\linewidth}
If there exist matchings $N_i\subseteq E'_i[U_s]$ for $i\in[k]$ such that $uv_1\in N_{c(uv_1)}$, $uv_2\in N_{c(uv_2)}$ and $U_s\subseteq V(\bigcup_{i=1}^k N_i)$, then $F\gets (F\setminus F[U_s])\cup(\bigcup_{i=1}^k N_i)$ and go to Step~\ref{step:spr}.\label{step:case2}\;        
\end{varwidth}
}
Take a maximum forest $F^\bullet$ in $F[U_s]$ and set $F\gets (F\setminus F[U_s])\cup F^\bullet$.\label{step:f''}{\footnotesize\color{Blue}{\tcp*[r]{Getting rid of parallel edges.}}}
\Return{$F$}\;
\end{algorithm}

\begin{thm}\label{thm:main}
Algorithm~\ref{algo:maxpf} provides a $5/9$-approximation for \maxpf in multigraphs in polynomial time.
\end{thm}
\begin{proof}
We first prove that Algorithm~\ref{algo:maxpf} constructs a feasible solution to \maxpf, and provide a lower bound on its size. Finally, we analyze the time complexity.

\paragraph{\textmd{\it Feasibility.}} We show that if the algorithm terminates, then it returns a properly colored forest $F$ in $G$. Throughout the algorithm, the edge set $F$ is the union of matchings $M_i\subseteq E_i$ for $i\in[k]$, hence it is properly colored. By Steps~\ref{step:f'} and~\ref{step:f''}, the algorithm outputs the union of a forest $F^\circ$ covering $U_r$ and a forest $F^\bullet$ covering $U_s$, which is a forest. These prove feasibility.

\paragraph{\textmd{\it Approximation factor.}} Let $F$, $U_s$, $U_r$ and $E'$ denote the corresponding sets at the termination of the algorithm, and let $G'\coloneqq (U,E')$ and $G''=(U, E[U]\setminus E')$. By Lemma~\ref{lemma:key}, we have 
\begin{equation}  
\opt[G] = \opt[G[U]] \le \opt[G'] + \opt[G'']. \label{eq:g'plusg''}
\end{equation} 
We give upper bounds on $\opt[G']$ and $\opt[G'']$ separately.

\begin{cl} \label{cl:g'}
$\opt[G'] = |F[U_s]| = |U_s|/2$.
\end{cl}
\begin{proof}
Clearly, $\opt[G']\ge |U_s|/2$ as the output of Algorithm \ref{algo:maxpf} has this many edges in $E[U_s]\subseteq E'$. 

Let $F'$ be a maximum-size properly colored forest of $G'$ that covers every vertex in $U_s$; note that such a forest exists by Lemma~\ref{lemma:key}. Suppose to the contrary that $|F'|>|U_s|/2$. Then, either there is an edge $e=uv\in F' \setminus E[U_s]$, or there are edges $e_1=uv_1$ and $e_2=uv_2$ with $c(e_1)\ne c(e_2)$ such that $e_1,e_2\in F'\cap E[U_s]$. In particular, there are matchings $N_1,\dots, N_k$ with $N_i\subseteq E'_i$ such that they either cover every vertex in $U_s+v$ and $uv\in N_{c(uv)}$, or they cover $U_s$ and $e_1\in N_{c(e_1)}$ and $e_2\in N_{c(e_2)}$. 
Both cases lead to a contradiction, since the algorithm would have found such matchings $N_1,\dots, N_k$ in Step~\ref{step:case1} or Step~\ref{step:case2}. Therefore, $\opt [G']=|U_s|/2$ indeed holds.
\end{proof}

We use the following simple observation to bound $\opt[G'']$.

\begin{cl} \label{cl:adj}
If an edge $e \in E[U] \setminus E'$ connects two components of $F$, then there exists an edge in $F[U_r]$ which is adjacent to $e$ and has the same color.
\end{cl}
\begin{proof}
Since $E[U_s] \subseteq E'$, $e$ has at least one endpoint in $U_r$.
If $e=vw$ such that $v \in U_s$ and $w\in U_r$, then $c(e) \in c(\delta_F(w))$ by $e \not\in E'$ and the definition of $E'$.
Otherwise, $e$ is spanned by $U_r$, and since it was not added to $F$ in Step~\ref{step:max}, it is adjacent to an edge of $F$ having the same color.
\end{proof}

With the help of the claim, we can bound $\opt[G'']$.

\begin{cl}  \label{cl:g''}
$\opt[G''] \le 3\cdot |F[U_r]|$.
\end{cl}
\begin{proof}
Let $F''$ be a maximum-size properly colored forest of $G''$.
For each edge $f \in F[U_r]$, $F''$ has at most two edges adjacent to $f$ having color $c(f)$. 
Then, Claim~\ref{cl:adj} implies that $F''$ has at most $2\cdot |F[U_r]|$ edges connecting different components of $F$. 
As $F''$ is a forest, it has at most $|F[U_r]|$ edges spanned by a component of $F[U_r]$, thus $|F''| \le 3\cdot |F[U_r]|$ follows.
\end{proof}

Using \eqref{eq:g'plusg''}, Claim~\ref{cl:g'}, and Claim~\ref{cl:g''}, we get \[\opt[G] \le \opt[G'] + \opt[G''] \le |F[U_s]| + 3\cdot  |F[U_r]|  = |F| + 2\cdot |F[U_r]|,\]
which yields 
\begin{equation} \label{eq:main}
|F| \ge \opt[G] - 2\cdot |F[U_r]|.
\end{equation}
Using that $|U| \ge \opt[G[U]] = \opt[G]$, we get 
\begin{equation} \label{eq:nbound}
2\cdot |F| = |U_s| + 2\cdot |F[U_r]| = |U| - |U_r| + 2\cdot  |F[U_r]| \ge \opt[G] - |U_r| + 2\cdot  |F[U_r]|.
\end{equation}
Since each component of $F[U_r]$ has size at least three, we have $|F[U_r]| \ge 2/3\cdot |U_r|$. Thus \eqref{eq:nbound} implies  
\begin{equation} \label{eq:comb}
8|F| \ge 4\cdot \opt[G] - 4 \cdot |U_r| + 8\cdot |F[U_r]| \ge 4 \cdot \opt[G] + 2\cdot |F[U_r]|.
\end{equation}
By adding \eqref{eq:main} and \eqref{eq:comb}, we obtain
\[9\cdot |F| \ge 5 \cdot  \opt[G],\]
proving the approximation factor.

\paragraph{\textmd{\it Time complexity.}}

By Remark~\ref{rem:alg}, each step of the algorithm can be performed in polynomial time, and the total number of for loops in Steps~\ref{step:for1} and~\ref{step:for2} is also clearly polynomial in the number of edges of the graph. Hence it remains to show that the algorithm makes polynomially many steps back to Step~\ref{step:spr}. This follows from the fact that whenever the algorithm returns to Step~\ref{step:spr}, a local improvement was found and so the sum $|U_s|+|\comp(F)|$ strictly decreases that can happen at most $2n$ times. This concludes the proof of the theorem.
\end{proof}

The analysis in Theorem~\ref{thm:main} is tight for $k$-edge-colored multigraphs if $k\geq 4$; see Figure~\ref{fig:tight1} for an example.

%%%%%%%%%%%%%%%%%%%%%%%%%%%%%%%%
\subsection{Simple graphs and multigraphs with small numbers of colors}
\label{sec:simple}
%%%%%%%%%%%%%%%%%%%%%%%%%%%%%%%%
While Algorithm~\ref{algo:maxpf} provides a $5/9$-approximation in general, the approximation factor can be improved if the graph is simple or the number of colors is small. In what follows, we show how to get better guarantees when $G$ is simple or $k\leq 3$.

\begin{thm} \label{thm:simple} 
Algorithm~\ref{algo:maxpf} provides a $4/7$-approximation for \maxpf in simple graphs and in 3-edge-colored multigraphs.
\end{thm}
\begin{proof}
We use the notation and extend the proof of Theorem~\ref{thm:main}.
Consider an instance where $G$ is simple or $k=3$;
this assumption is in fact used only in the next simple observation.

\begin{cl} \label{cl:spec}
Let $C$ be a component of $F$ with $|C|=3$.
If $|F''[C]| = 2$, then there exist $e \in F''[C]$ and $f \in F[U_r]$ such that $c(e)=c(f)$ and $e$ and $f$ has at least one common endpoint.
\end{cl}
\begin{proof} 
If $G$ is simple, then $|E[C]| \le 3$, thus $|F''[C] \cap F[U_r]| \ge 1$.
If $k=3$, then $|c(F''[C]) \cap c(F[C])| \ge 1$, that is, $c(e)=c(f)$ for some $e \in F''[C]$ and $f \in F[C]$. Since $|C|=3$, $e$ and $f$ has at least one common endpoint.
\end{proof}

Let $m_3 \coloneqq |\{C \in \comp(F) \mid |C| = 3\}|$. Using Claim~\ref{cl:spec}, we strengthen Claim~\ref{cl:g''} as follows.

\begin{cl} \label{cl:g''spec} $\opt[G''] \le 3\cdot |F[U_r]|-m_3$.
\end{cl}
\begin{proof}
Let $\gamma \coloneqq |\{C \in \comp(F) \mid |C| = 3, |F''[C]| = 2\}|$.
Let $F''_1$ denote the set of edges $uv \in F''$ such that $u$ and $v$ are in different components of $F$, and let $F''_2 \coloneqq F'' \setminus F''_1$.
Claim~\ref{cl:adj} and Claim~\ref{cl:spec} imply that $F''$ has at least $|F''_1| + \gamma$ edges $e$ for which there exists $f \in F[U_r]$ such that $c(e)=c(f)$ and $e$ and $f$ has at least one common endpoint.
For each $f \in F[U_r]$, $F''$ has at most two edges having the same color as $f$ and sharing at least one common endpoint with $f$, implying $2\cdot |F[U_r]| \ge |F''_1| + \gamma$.
Since $F$ has $m_3-\gamma$ size-three components spanning at most one edge of $F''_2$, we have $|F''_2| \le |F[U_r]|-(m_3-\gamma)$.
Then, \[|F''| = |F''_1| + |F''_2| \le (2\cdot |F[U_r]-\gamma) + (|F[U_r]+\gamma-m_3) = 3\cdot |F[U_r]| - m_3,\]
and the claim follows.
\end{proof}

Using \eqref{eq:g'plusg''}, Claim~\ref{cl:g'} and Claim~\ref{cl:g''spec}, we get
\[\opt[G] \le \opt[G'] + \opt[G''] \le |F[U_s]| + 3\cdot |F[U_r]| - m_3 = |F| + 2\cdot |F[U_r]| - m_3,\]
that is, 
\begin{equation} \label{eq:mainspec}
|F| \ge \opt[G] - 2\cdot |F[U_r]| + m_3.
\end{equation}
Since $F[U_r]$ has $m_3$ components of size three and the other components of $F[U_r]$ has size at least four, we have $|F[U_r]| \ge 2 m_3 + 3/4 \cdot (|U_r| - 3m_3) = 3/4 \cdot |U_r| - 1/4 \cdot m_3$.
Then, \eqref{eq:nbound} implies  
\begin{equation} \label{eq:combspec}
6\cdot |F| \ge 3\cdot \opt[G] - 3\cdot  |U_r| + 6\cdot |F[U_r]| \ge 3 \cdot \opt[G] + 2\cdot |F[U_r]| -m_3.
\end{equation}
By adding \eqref{eq:mainspec} and \eqref{eq:combspec}, we obtain
\[7\cdot |F| \ge 4 \cdot  \opt[G],\]
proving the approximation factor.
\end{proof}

The analysis in Theorem~\ref{thm:simple} is tight for $3$-edge-colored multigraphs and for $k$-edge-colored simple graphs for $k\geq 4$; see Figures~\ref{fig:tight2} and~\ref{fig:tight3} for examples.

\begin{figure}[t!]
\centering
\begin{subfigure}{0.48\textwidth}
    \centering
    \includegraphics[width=0.85\textwidth]{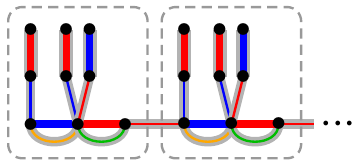}
    \caption{An example showing that the approximation factor of $5/9$ in Theorem~\ref{thm:main} is tight for 4-edge-colored multigraphs. If the  graph consists of $\ell$ blocks, then the approximation ratio is $5\ell/(9\ell-1)$.}
    \label{fig:tight1}
\end{subfigure}\hfill
\begin{subfigure}{0.48\textwidth}
    \centering
    \includegraphics[width=0.68\textwidth]{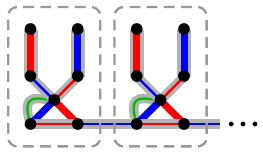}
    \caption{An example showing that the approximation factor of $4/7$ in Theorem~\ref{thm:simple} is tight for 3-edge-colored multigraphs. If the  graph consists of $\ell$ blocks, then the approximation ratio is $4\ell/(7\ell-1)$.}
    \label{fig:tight2}
\end{subfigure}

\vspace{0.5cm}

\begin{subfigure}{0.48\textwidth}
    \centering
    \includegraphics[width=0.85\textwidth]{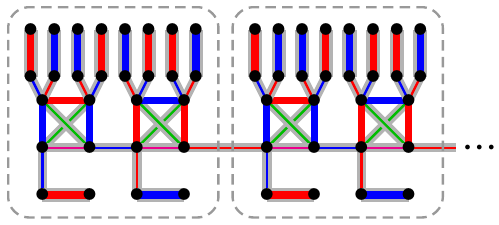}
    \caption{An example showing that the approximation factor of $4/7$ in Theorem~\ref{thm:simple} is tight for 4-edge-colored simple graphs. If the  graph consists of $\ell$ blocks, then the approximation ratio is $16\ell/(28\ell-1)$.}
    \label{fig:tight3}
\end{subfigure}\hfill
\begin{subfigure}{0.48\textwidth}
    \centering
    \includegraphics[width=0.73\textwidth]{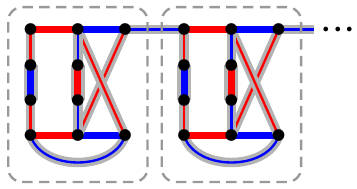}
    \caption{An example showing that the approximation factor of $3/5$ in Theorem~\ref{thm:multi2} is tight for 2-edge-colored graphs. If the  graph consists of $\ell$ blocks, then the approximation ratio is $6\ell/(10\ell-1)$.}
    \label{fig:tight4}
\end{subfigure}
\caption{Tight examples for Algorithm~\ref{algo:maxpf} in different settings. Thick edges denote the properly colored forest found by the algorithm, while edges with a grey outline form an optimal solution. The graphs are obtained by repeating the blocks enclosed by the dashed boxes $\ell$ times.}
\label{figh:tight}
\end{figure}

\begin{thm} \label{thm:multi2}
Algorithm~\ref{algo:maxpf} provides a $3/5$-approximation for \maxpf in 2-edge-colored multigraphs.
\end{thm}
\begin{proof}
We use the notation and extend the proof of Theorem~\ref{thm:main} assuming that $k=2$. For $e \in F''$, define 
\[x(e) \coloneqq |\{f \in F[U_r] \mid c(e)=c(f), \text{ $e$ and $f$ has at least one common endpoint}\}|.\]
For a subset $S \subseteq F''$, we use the notation $x(S) \coloneqq \sum_{e \in S} x(e)$.

\begin{cl} \label{cl:adjk2} $x(F''[C]) \ge |F''[C]|-1$ for every even component $C \in \comp(F[U_r])$, and $x(F[C]) \ge |F''[C]|$ for every odd component $C \in \comp(F[U_r])$.
\end{cl}
\begin{proof}
Let $\ell \coloneqq |C|$. Since $k=2$, $F[C]$ is an alternating path, let $v_1, v_2, \dots, v_\ell$ denote its vertices and $f_1, \dots, f_{\ell-1}$ denote its edges such that $f_i = v_iv_{i+1}$ for $i \in [\ell-1]$.
For each edge $e \in F''[U_r]$ we have $x(e) \ge 1$ unless $e = v_1 v_\ell$ and $c(e) \ne c(f_1) = c(f_{\ell-1})$. This proves the claim since $c(f_1) \ne c(f_{\ell-1})$ if $\ell = |C|$ is odd.
\end{proof}

Let $m_3 \coloneqq |\{C \in \comp(F) \mid |C| = 3\}|$.  Using Claim~\ref{cl:adjk2}, we strengthen Claim~\ref{cl:g''} as follows.

\begin{cl} \label{cl:g''k2}
$\opt[G''] \le |F[U_r]| + |U_r| + m_3$.
\end{cl}
\begin{proof}
$F[U_r]$ has $|U_r|-|F[U_r]|$ components, thus it has at most $|U_r|-|F[U_r]|-m_3$ even components.
Claim~\ref{cl:adj} implies that $x(e) \ge 1$ holds for each edge $e \in F''$ connecting two components of $F$.
Using Claim~\ref{cl:adjk2}, it follows that $x(F'') \ge |F''|-(|U_r|-|F[U_r]|-m_3)$. For each edge $f \in F[U_r]$, $F''$ has at most two edges having the same color as $f$ and at least one common endpoint of $f$, thus $x(F'') \le 2|F[U_r]|$. Then,
\[|F''| \le x(F'') + |U_r| - |F[U_r]| - m_3 \le |F[U_r]| + |U_r| - m_3.\]
\end{proof}
Using \eqref{eq:g'plusg''}, Claim~\ref{cl:g'} and Claim~\ref{cl:g''k2}, we get
\[\opt[G] \le \opt[G'] + \opt[G''] \le |F[U_s]| + |F[U_r]| +|U_r| - m_3 = |F| + |U_r| - m_3,\]
that is, 
\begin{equation} \label{eq:maink2}
|F| \ge \opt[G] - |U_r| + m_3.
\end{equation}
As in the proof of Theorem~\ref{thm:simple},  $|F[U_r]\ge 3/4 \cdot |U_r| - 1/4 \cdot m_3$, thus 
\eqref{eq:nbound} implies  
\begin{equation} \label{eq:combk2}
4\cdot |F| \ge 2\cdot \opt[G] -  2|U_r| + 4|F[U_r]| \ge 2\cdot \opt[G] + |U_r| -m_3.
\end{equation}
By adding \eqref{eq:maink2} and \eqref{eq:combk2}, we get \[5 \cdot |F| \ge 3 \cdot \opt[G],\]
proving the approximation factor.
\end{proof}

The analysis in Theorem~\ref{thm:multi2} is tight for $2$-edge-colored multigraphs; see Figure~\ref{fig:tight4} for an example.

%%%%%%%%%%%%%%%%%%%%%%%%%%%%%%%%
\subsection{Simple graphs with small numbers of colors}
\label{sec:small}
%%%%%%%%%%%%%%%%%%%%%%%%%%%%%%%%

For simple graphs, the algorithm can be significantly simplified while leading to even better approximation factors if the number of colors is small. The modified algorithm is presented as Algorithm~\ref{alg-2-colors}. First, we consider the case $k=2$.

\begin{algorithm}[t!]
\caption{Approximation algorithm for \maxpf in simple graphs.}\label{alg-2-colors}
\DontPrintSemicolon

\KwIn{A simple graph $G=(V,E)$ with edge-coloring $c\colon E\to[k]$.}
\KwOut{A properly colored forest $F$ in $G$.}

\medskip

Find maximum matchings $M_i\subseteq E_i$ for $i\in[k]$ maximizing $|\bigcup_{i=1}^k V(M_i)|$.\label{step:1}\;
Let $F'\coloneqq \bigcup_{i=1}^k M_i$.\label{step:2}\;
Take a maximum forest $F$ in $F'$.\label{step:3}\;
\Return{$F$}\;
\end{algorithm} 

\begin{thm}\label{thm:34}
Algorithm \ref{alg-2-colors} provides a $\frac{3}{4}$-approximation for \maxpf in 2-edge-colored simple graphs in polynomial time.
\end{thm}
\begin{proof}
Let $M_1$ and $M_2$ denote the maximum matchings found in Step~\ref{step:1} of the algorithm. Then in Step~\ref{step:2}, $F'$ is a properly colored edge set which is the vertex-disjoint union of paths and even cycles. As the graph is simple, every cycle has length at least $4$. In Step~\ref{step:3}, we delete no edge from the paths and exactly one edge from each cycle. Since every cycle had length at least $4$, we deleted at most $1/4\cdot(|M_1|+|M_2|)$ edges and hence the algorithm outputs a solution of size $|F|\geq 3/4\cdot(|M_1|+|M_2|)$. On the other hand, $\opt[G]\leq |M_1|+|M_2|$ clearly holds since every properly colored forest of $G$ decomposes into the union of a matching in $E_1$ and a matching in $E_2$. This concludes the proof of the theorem.
\end{proof}

The analysis in Theorem~\ref{thm:34} is tight for $2$-edge-colored simple graphs; see Figure~\ref{fig:tight6} for an example.

\begin{rem}
Note that the proof of Theorem~\ref{thm:34} only uses that $M_1$ and $M_2$ are maximum matchings and does not rely on the fact that $|V(M_1\cup M_2)|$ is maximized.     
\end{rem}

Now we discuss the case when $k=3$.

\begin{thm}\label{thm:58}
Algorithm \ref{alg-2-colors} provides a $\frac{5}{8}$-approximation for \maxpf in 3-edge-colored simple graphs in polynomial time.
\end{thm}
\begin{proof}
Let $M_1,M_2$ and $M_3$ denote the maximum matchings found in Step~\ref{step:1} of the algorithm. Then in Step~\ref{step:2}, $F'$ is a properly colored edge set in which every vertex has degree at most $3$.

\begin{cl}\label{cl:two}
$|F'(C)|= 1$ for every component $C\in\comp(F')$ of size $2$.
\end{cl}
\begin{proof}
The statement follows by the assumption that $G$ is simple.
\end{proof}

\begin{cl}\label{cl:even}
$|F'(C)|\leq 3/2\cdot |C|$ for every even component $C\in\comp(F')$.
\end{cl}
\begin{proof}
The statement follows from the fact that each vertex has degree at most $3$ in $F'$.
\end{proof}

\begin{cl}\label{cl:odd}
$|F'(C)|\leq 3/2\cdot (|C|-1)$ for every odd component $C\in\comp(F')$.
\end{cl}
\begin{proof}
Suppose to the contrary that $|F'(C)|>3/2\cdot (|C|-1)$. Since every vertex has degree at most $3$ in $F'$, $C$ either contains $|C|-2$ vertices of degree $3$ and two vertices of degree $2$, or $|C|-1$ vertices of degree $3$ and one vertex $u$ of degree at least one in $F'$. However, the former case cannot happen as $C$ is an odd component and the sum of the degrees of the vertices is an even number, namely $2|F'|$. Let $e\in F'$ be an edge incident to $u$. Since every vertex in $C-u$ has degree exactly $3$, each vertex in $C$ is incident to an edge of color $c(e)$ in $F'$. However, $F'$ is a properly colored edge set, hence the edges in $F'(C)$ colored by $c(e)$ form a perfect matching of $C$, contradicting $|C|$ being odd. 
\end{proof}

For $i\in[n]$, let $m_i$ denote the number of components in $\comp(F')$ containing $i$ vertices. Furthermore, let $m\coloneqq\sum_{i=2}^{\lfloor n/2\rfloor} m_{2i}$, that is, $m$ is number of even components in $F'$ of size at least four. 
Using Claim~\ref{cl:two}, Claim~\ref{cl:even} and Claim~\ref{cl:odd},
we get
\begin{align*}
2\cdot |F'| &\le \sum_{\substack{C \in \comp(F')\\|C|=2}} 2 + \sum_{\substack{C \in \comp(F')\\ \text{$C$ is even}\\ |C| \ge 4}} 3\cdot |C| + \sum_{\substack{C \in \comp(F')\\ \text{$C$ is odd}}} 3\cdot (|C|-1) \\
& = \sum_{C \in \comp(F')}3\cdot (|C|-1) - m_2 + 3m \\
& = 3\cdot |F| - m_2 + 3m.
\end{align*}
Note that $\opt[G]\leq |M_1|+|M_2|+|M_3| = |F'|$ clearly holds since every properly colored forest of $G$ decomposes into the union of a matching in $E_1$, a matching in $E_2$, and a matching in $E_3$. Then, by rearranging the previous inequality, we get
\begin{equation} \label{eq:main3col}
3\cdot  |F| \ge 2\cdot \opt[G] + m_2 - 3m.
\end{equation}

Let $U \coloneqq \bigcup_{i=1}^3 V(M_i)$. Since each matching-coverable set can be covered by maximum matchings, $U$ is a maximum-size matching-coverable set, thus $\opt[G] = \opt[G[U]]$ holds by Lemma~\ref{lemma:key}.
Now $F$ is a forest, thus
$|F| = |F[U]| = |U| - \sum_{i=2}^n m_i = |U| - m_2 - m - \sum_{j=1}^{\lfloor (n-1)/2\rfloor} m_{2j+1}$, that is, $\sum_{j=1}^{\lfloor (n-1)/2\rfloor} m_{2j+1} = |U| - m_2 - m - |F|$.
Using this equation and the fact that $U$ is the union of the components of $F$ with size at least two, we have
\begin{align*}
2\cdot |U| & = 2\cdot \sum_{i=2} i \cdot m_i\\
& \ge 4 m_2 + 8m + 6 \cdot \sum_{j=1}^{\lfloor (n-1)/2\rfloor} m_{2j+1} \\
& \ge 4 m_2 + 8m + 5 \cdot \sum_{j=1}^{\lfloor (n-1)/2\rfloor} m_{2j+1} \\
& = 4m_2 + 8m + 5 \cdot (|U| - m_2 - m - |F|) \\
& =5\cdot |U| - m_2 + 3 m - 5\cdot  |F|.
\end{align*}
Rearranging and using $|U| \ge \opt[G[U]] = \opt[G]$, we obtain
\begin{equation} \label{eq:other3col}
5\cdot |F| \ge 3 \cdot \opt[G] - m_2 + 3m.
\end{equation}
By adding \eqref{eq:main3col} and \eqref{eq:other3col}, we get
\[8 \cdot |F| \ge 5 \cdot \opt[G],\]
proving the approximation factor.
\end{proof}

The analysis in Theorem~\ref{thm:58} is tight for $3$-edge-colored simple graphs; see Figure~\ref{fig:tight5} for an example.

\begin{figure}[t!]
\centering
\begin{subfigure}{0.48\textwidth}
    \centering
    \includegraphics[width=0.7\textwidth]{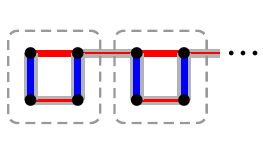}
    \caption{An example showing that the approximation factor of $3/4$ in Theorem~\ref{thm:34} is tight for 2-edge-colored simple graphs. If the  graph consists of $\ell$ blocks, then the approximation ratio is $3\ell/(4\ell-1)$.}
    \label{fig:tight6}
\end{subfigure}\hfill
\begin{subfigure}{0.48\textwidth}
    \centering
    \includegraphics[width=0.68\textwidth]{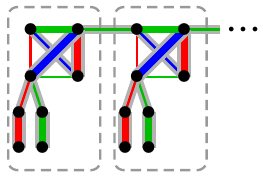}
    \caption{An example showing that the approximation factor of $5/8$ in Theorem~\ref{thm:58} is tight for 3-edge-colored simple graphs. If the  graph consists of $\ell$ blocks, then the approximation ratio is $5\ell/(8\ell-1)$.}
    \label{fig:tight5}
\end{subfigure}
\caption{Tight examples for Algorithm~\ref{alg-2-colors} for $k=2$ and $3$. Thick edges denote the properly colored
forest found by the algorithm, while edges with a grey outline form an optimal solution. The graphs are obtained by repeating the blocks enclosed by the dashed boxes $\ell$ times.}
\label{figh:tight_second}
\end{figure}

\begin{rem}
A key ingredient of Algorithm~\ref{alg-2-colors} is that it starts with maximum matchings $M_i\subseteq E_i$, which makes it possible to compare the size of the solution output by the algorithm against $\opt\leq\sum_{i=1}^k|M_i|$. In contrast, Algorithm~\ref{algo:maxpf} starts with arbitrary matchings $M_i\subseteq E_i$ maximizing $|V(\bigcup_{i=1}^k M_i)|$. The reason why that algorithm operates with matchings instead of maximum matchings is that in certain steps we need to find matchings containing some fixed edges, hence they cannot necessarily be chosen to be maximum matchings.
\end{rem}

%%%%%%%%%%%%%%%%%%%%%%%%%%%%%%%%
\subsection{Approximating \maxpt}
\label{sec:maxptapx}
%%%%%%%%%%%%%%%%%%%%%%%%%%%%%%%%

Finally, for any $\varepsilon >0$ we give an $1/\sqrt{(2+\varepsilon )(n-1)}$-approximation algorithm for \maxpt in complete multigraphs. The approximation factor is far from being constant; still, the algorithm is of interest since its approximation guarantee is better than the general upper bound on the approximability of \maxpt. %In our proof, we rely on the following result of Borozan et al.~\cite{borozan2019maximum}. 

Our algorithm for \maxpt in complete multigraphs is presented as Algorithm~\ref{algo:maxpt}.

\begin{thm}\label{thm:maxptalg}
For complete multigraphs on $n$ vertices and for any fixed constant $\varepsilon >0$, Algorithm \ref{algo:maxpt} provides a $1/\sqrt{(2+\varepsilon)(n-1)}$-approximation for \maxpt in polynomial time. 
\end{thm}
\begin{proof}
First, we show that deleting edges in Step~\ref{step:delete} does not decrease the size of the optimal solution. Indeed, for any optimal solution $F_{opt}$, if $e=vw\in F_{opt}$ but $e$ is deleted, then there are at least $n$ parallel edges between $v$ and $w$ having different colors. As the degrees of $v$ and $w$ are at most $n-1$ in $F_{opt}$, there is always at least one edge $f$ among those parallel ones such that $F_{opt}-e+f$ is a properly colored tree again. Note that after the deletion of unnecessary parallel edges, the total number edges of the graph is bounded by $n^3$.

Let $\varepsilon >0$ be the parameter of the algorithm. If $n<n_{\varepsilon}=(\varepsilon^2+9\varepsilon+18)/\varepsilon^2$, then the output is clearly optimal. Furthermore, the number of possible solutions is bounded by $\binom{n_\varepsilon^3}{n_\varepsilon}$ which is a constant, hence the runtime is constant. 

Assume now that $n\ge n_{\varepsilon}$. Let $V_1\cup V_2$ be the partition of $V$ as in Theorem \ref{thm:maxpt-partition}. We may assume that $V_1,V_2\neq\emptyset$ since otherwise the algorithm clearly gives an optimal solution. Let $F_1$ and $F_2$ be maximum-size properly colored trees in $G[V_1]$ and $G[V_2]$, respectively. Let $n_1\coloneqq |V_1|$, $n_2\coloneqq |V_2|$ and $x_1\coloneqq \opt[G[V_1]]=|F_1|=n_1-1$, $x_2\coloneqq \opt[G[V_2]]=|F_2|$. The forest $F_{12}$ in Step~\ref{step:bm} can be determined using a maximum bipartite matching algorithm in a bipartite graph $H=(S,T;W)$ defined as follows. The vertex set $S$ contains a vertex $(v,i)$ for each $v\in V_1$ and color $i\in[k]$ such that $v$ has no incident edges in $F_1$ having color $i$, that is, $S=\{(v,i)\mid v\in V_1,i\notin c(\delta_{F_1}(v))\}$. The vertex set $T$ contains a copy of each vertex in $V_2$, that is, $T=\{v\mid v\in V_2\}$. Finally, there is an edge added between $(v,i)\in S$ and $u\in T$ in $W$ if $uv\in E$ has color $i$ in $G$. It is not difficult to check that a maximum matching of $H$ gives a properly colored forest $F_{12}$ that can be added to $F_1$ and with respect to that, covers as many vertices in $V_2$ as possible. 

For the output $F$ of Algorithm~\ref{algo:maxpt}, either we have $|F|=x_2$ or $|F|=x_1+y$, where  $y=|F_{12}|$. Recall that $n_1\geq 1$ by our assumption, hence $x_1+y\geq 1$. Indeed, this clearly holds if $n_1\geq 2$, while if $n_1=1$ then $y\geq 1$ by the completeness of the multigraph. Let $F_{opt}$ be an optimal properly colored tree in $G$. 
%We may assume that $F_{opt}$ covers at least one vertex in $V_1$, since otherwise $|F_{opt}| = x_2$ and the algorithm finds an optimal solution. 
We claim that $\opt[G]=|F_{opt}|\le 3x_1+ y +(2x_1+y)x_2$. To see this, let $U\coloneqq\{u\in V_2\mid\text{there exists $uv\in F_{opt}$ with $v\in V_1$}\}$ and set $U'\coloneqq\{u\in U\mid \text{$c(uv)\in c(\delta_{F_1}(v))$ for every $uv\in F_{opt}$ with $v\in V_1$}\}$. Since every edge of $F_1$ is adjacent to at most two edges in $F_{opt}$ having the same color, we have $|U'|\le 2x_1$. Moreover, by the choice of $F_{12}$, we have $|U\setminus U'|\le y$. These together imply $|U|\le 2x_1+y$. 
Now $F_{opt}\setminus E[V_1\cup U]$ is the union of properly colored trees in $G[V_2]$, all of which can have size at most $x_2$. By the above, there are at most $|U|=2x_1+y$ such components as $F_{opt}$ is connected, leading to $|F_{opt} \setminus E[V_1\cup U]| \le (2x_1+y)x_2$. Finally, observe that $|F_{opt}\cap E[V_1\cup U]| \le 3x_1+y$ by $|V_1|\leq x_1+1$ and $|U|\leq 2x_1+y$. Since $F_{opt}$ has at most $|V|-1$ edges, these together show $\opt [G] =|F_{opt}|\le \min\{n-1,3x_1+y+(2x_1+y)x_2\}$. 

The approximation factor of Algorithm~\ref{algo:maxpt} is hence at least
$\max\{x_1+y, x_2\}/\min\{3x_1+y+(2x_1+y)x_2, n-1\}$. To lower bound this expression, let $x'_1\coloneqq x_1+y$. Then, it suffices to show that  
\[
\frac{\max\{x'_1,x_2\}}{\min\{n-1,3x'_1+2x'_1x_2\}} \ge \frac{1}{\sqrt{(2+\varepsilon )(n-1)}}
\]
for $1\leq x'_1\leq n-1$ and $0\le x_2\le n-1$, since the value on the left hand side is a lower bound on the approximation factor. Assume that this is not the case, and in particular, we have $x_2< \sqrt{n-1}/\sqrt{2+\varepsilon}$ and $x'_1/(3x'_1+2x'_1x_2)<1/\sqrt{(2+\varepsilon)(n-1)}$ for some $n\geq n_\varepsilon$. From the latter inequality, we get $\sqrt{(2+\varepsilon)(n-1)}/2-3/2<x_2$. Therefore, $\sqrt{(2+\varepsilon)(n-1)}/2-3/2<x_2<\sqrt{n-1}/\sqrt{2+\varepsilon}$. However, $\sqrt{(2+\varepsilon)(n-1)}/2-3/2\ge \sqrt{n-1}/\sqrt{2+\varepsilon}$ whenever $n\ge n_{\varepsilon}$, a contradiction.

We conclude that Algorithm~\ref{algo:maxpt} provides a $1/\sqrt{(2+\varepsilon )(n-1)}$-approximation. Also, by Theorem \ref{thm:maxpt-partition} and the fact that a maximum-size matching can be computed in polynomial time, the running time is polynomial. This concludes the proof of the theorem.
\end{proof}

\begin{algorithm}[t!] 
\caption{Approximation algorithm for \maxpt in complete multigraphs.}\label{algo:maxpt}
\DontPrintSemicolon

\KwIn{A complete multigraph $G=(V,E)$ with edge-coloring $c\colon E\to [k]$ and $\varepsilon>0$.}
\KwOut{A properly colored tree $F$ in $G$ such that $|F|\ge \opt /\sqrt{(2+\varepsilon)(n-1)} $.}

\medskip

\If{
$\exists v,w\in V,|E[\{v,w\}]|\geq n$
}{
Choose $n$ parallel edges between $v$ and $w$ arbitrarily and delete the remaining ones.\label{step:delete}\;
}
Let $F\coloneqq\emptyset$ and $n_\varepsilon\coloneqq (\varepsilon^2+9\varepsilon +18)/\varepsilon^2 $.\;
\If{ $n< n_{\varepsilon }$}{
Compute all properly colored trees in $G$ and let $F_{opt}$ be one with maximum size.\label{step:brute}\;
$F\gets F_{opt}$\;
}
\Else{
%Let $E_i\coloneqq\{e\in E\mid c(e)=i\}$ for $i\in[k]$. 
Compute $V_1,V_2$ and optimal properly colored tree $F_i$ of $G[V_i]$ for $i\in[2]$ as in Theorem \ref{thm:maxpt-partition}.\;
Let $E' \coloneqq \{ vw \mid v\in V_1, w\in V_2, c(vw)\notin c(\delta_{F_1}(v))\}$.\;
\begin{varwidth}{0.9\linewidth}
    Compute a properly colored forest $F_{12}\subseteq E'$ that covers a maximum number of vertices in $V_2$ and $|\delta_{F_{12}}(v)|\le 1$ for each $v\in V_2$.\label{step:bm}\;
\end{varwidth}
\If{$|F_1|+|F_{12}|\ge |F_2|$}{
$F\gets F_1\cup F_{12}$\;
}
\Else{
$F\gets F_2$\;
}
}
\Return{$F$}\;
\end{algorithm}

\begin{rem}
For $\varepsilon = 2$, the algorithm provides a $1/(2\sqrt{n-1})$-approximation and $n_2=10$. That is, the brute force approach of Step~\ref{step:brute} is only executed for $n\le 9$. However, any properly colored tree with two edges gives a $1/(2\sqrt{n-1})$-approximation, and deciding the existence of such a tree requires $\binom{|E|}{2}$ steps. 
\end{rem}

%%%%%%%%%%%%%%%%%%%%%%%%%%%%%%%%
\section{Conclusions}
\label{sec:conclusions}
%%%%%%%%%%%%%%%%%%%%%%%%%%%%%%%%

In this paper we introduced and studied the Maximum-size Properly Colored Forest problem, in which we are given an edge-colored undirected graph and the goal is to find a properly colored forest of maximum size. We showed that the problem is closely related to fundamental problems of combinatorial optimization such as the Bounded Degree Spanning Tree, the Bounded Degree Matroid, the Multi-matroid Intersection, and the (1,2)-Traveling Salesman problems. We considered the problem for complete and non-complete, simple and non-simple graphs, and presented polynomial-time approximation algorithms as well as inapproximability results depending on the number of colors. 

We close the paper by mentioning some open problems:

\begin{enumerate}\itemsep0em
    \item The probably most straightforward question is closing the gap between the lower and upper bounds on the approximability of the problem. Our results on the inapproximability of the problem provide only very weak (close to 1) upper bounds. Providing significantly smaller upper bounds would be a significant step towards getting an idea of the exact values.
    \item The weighted variant of \maxpf can be defined in a straightforward manner, where the goal is to find a properly colored forest of maximum total weight. While some of the results, e.g. Theorem~\ref{thm:34} can be extended to the weighted setting as well, this is not always true. A systematic study of the problem assuming edge weights is therefore of interest. 
    \item The algorithms of \cite{kiraly2012degree} and \cite{linhares2020approximate} both rely on iteratively solving a corresponding LP and then fixing variables having $0$ or $1$, with additional ideas for relaxing the constraints which lead to an approximate solution. An interesting question is whether such an approach can be used for approximating the maximum size or maximum weight of a properly colored forest in our setting.
\end{enumerate}

\medskip
%%%%%%%%%%%%%%%%%%%%%%%%%%%%%%%%
\paragraph{Acknowledgement.}
%%%%%%%%%%%%%%%%%%%%%%%%%%%%%%%%
Yuhang Bai was supported by the National Natural Science Foundation of China -- NSFC, grant numbers 12071370 and 12131013 and by Shaanxi Fundamental Science Research Project for Mathematics and Physics -- grant number 22JSZ009. Gergely K\'al Cs\'aji was supported by the Hungarian Scientific Research Fund, OTKA, Grant No. K143858 and by the Doctoral Student Scholarship Program (number C2258525) of the Cooperative Doctoral Program of the Ministry of Innovation and Technology financed by the National Research, Development, and Innovation Fund.
Tamás Schwarcz was supported by the \'{U}NKP-23-3 New National Excellence Program of the Ministry for Culture and Innovation from the source of the National Research, Development and Innovation Fund. This research has been implemented with the support provided by the Lend\"ulet Programme of the Hungarian Academy of Sciences -- grant number LP2021-1/2021, by the Ministry of Innovation and Technology of Hungary from the National Research, Development and Innovation Fund, financed under the ELTE TKP 2021-NKTA-62 funding scheme, and by Dynasnet European Research Council Synergy project (ERC-2018-SYG 810115).

%%%%%%%%%%%%%%%%%%%%%%%%%%%%%%%%
\bibliographystyle{abbrv}
\bibliography{maxpf}
    
\end{document}